\documentclass[10pt]{article}
\usepackage[OT1]{fontenc}
\pdfoutput=1
\usepackage[dvipsnames]{xcolor}
\usepackage{amsthm,amsmath,amsbsy,amssymb,array,bm,graphicx,epstopdf,mathtools,calc,xfrac}
\usepackage[numbers,sort&compress]{natbib}
\usepackage{xr-hyper}
\RequirePackage[activate={true,nocompatibility},final,tracking=true,kerning=true,spacing=true,factor=1100,stretch=10,shrink=10,spacing=basictext]{microtype}
\microtypecontext{spacing=nonfrench}
\RequirePackage[colorlinks=true,citecolor=MidnightBlue,urlcolor=MidnightBlue,linkcolor=PineGreen,breaklinks]{hyperref}
\usepackage{url} 
\usepackage[ruled,linesnumbered]{algorithm2e}
\SetKwInput{KwInput}{Input}
\SetKwInput{KwOutput}{Output}
\usepackage{subfigure,subdepth,mdframed,tikz,pgf,pgffor}
\usepgfmodule{shapes,plot}
\usetikzlibrary{decorations,arrows,positioning,calc}
\xdefinecolor{darkgreen}{RGB}{0,193,0}
 
 \newlength{\mylength}
\allowdisplaybreaks
\makeatletter
\newcommand{\pushright}[1]{\ifmeasuring@#1\else\omit\hfill$\displaystyle#1$\fi\ignorespaces}
\newcommand{\pushleft}[1]{\ifmeasuring@#1\else\omit$\displaystyle#1$\hfill\fi\ignorespaces}
\makeatother
\let\originalleft\left%
\let\originalright\right%
\renewcommand{\left}{\mathopen{}\mathclose\bgroup\originalleft}%
\renewcommand{\right}{\aftergroup\egroup\originalright}%
\makeatletter%
\def\mybigx#1{\dimen@#1\relax%
\mathchoice%
{\vbox to \dimen@{}}%
{\vbox to \dimen@{}}%
{\vbox to .7\dimen@{}}%
{\vbox to .5\dimen@{}}}%
\def\mybig#1{{\hbox{$\left#1\mybigx{0.8em}\right.\n@space$}}}%
\makeatother%
\usepackage[titletoc,title]{appendix}
\usepackage{subfigure}
\usepackage{lipsum} 
\usepackage{a4wide}
\usepackage{float}
\usepackage{xcolor}
\usepackage{amsthm}
\usepackage{graphicx}
\usepackage[latin1]{inputenc}
\usepackage{multirow} 
\usepackage{bm}
\usepackage{epsfig}
\usepackage{lscape}
\usepackage{bm}
 \renewcommand{\P}{\ensuremath{\operatorname{\mathbb{P}}}}
 \newcommand{\E}{\ensuremath{\operatorname{\mathbb{E}}}}
 
 \newcommand{\var}{\ensuremath{\operatorname{Var}}}
 \newcommand{\cov}{\ensuremath{\operatorname{Cov}}}

\DeclareFontFamily{U}{mathx}{\hyphenchar\font45}
\DeclareFontShape{U}{mathx}{m}{n}{<-> mathx10}{}
\DeclareSymbolFont{mathx}{U}{mathx}{m}{n}
\usepackage{thmtools}
\usepackage{thm-restate}
\declaretheorem[name=Theorem]{thm}

\usepackage[noabbrev,capitalise,nameinlink]{cleveref}
\numberwithin{equation}{section}
\theoremstyle{plain}
\crefname{figure}{Fig.}{Figs.}
\crefname{table}{Table}{Tables}
\newtheorem{Definition}{Definition}
\crefname{Definition}{Definition}{Definition}
\crefname{defin}{Definition}{Definition}

\crefname{Lemma}{Lemma}{Lemmas}
\newtheorem{Proposition}{Proposition}
\crefname{Proposition}{Proposition}{Propositions}

\crefname{Theorem}{Theorem}{Theorems}
\crefname{thm}{Theorem}{Theorems}

\crefname{Corollary}{Corollary}{Corollaries}
\theoremstyle{remark}
\newtheorem{Remark}{Remark}
\crefname{Remark}{Remark}{Remarks}

\begin{document}
 \title{\bf Nonparametric regression for multiple heterogeneous networks}
 \author{Swati Chandna 
\\
Department of Economics, Mathematics and Statistics, \\
Birkbeck, University of London, UK\\
 and \\
 Pierre-Andre Maugis 
\\
Department of Statistical Science, \\
University College London, UK}
\date{}
 \maketitle
\begin{abstract}
We study nonparametric methods for the setting where multiple distinct networks are observed on the same set of nodes.
Such samples may arise in the form of replicated networks drawn from a common distribution, or in the form of
heterogeneous networks,
with the network generating process varying from one network to another,
e.g.~dynamic and cross-sectional networks.
Nonparametric methods for undirected networks have focused on estimation of the graphon model. 
While the graphon model accounts for nodal heterogeneity, it does not account for network heterogeneity,
a feature specific to applications where multiple networks 
are observed.
To address this setting of multiple networks,
we propose a \textit{multi-graphon} model which allows node-level as well as network-level heterogeneity.  
We show how information from multiple networks can be leveraged to enable estimation of the multi-graphon via standard nonparametric regression techniques, e.g.~kernel regression, 
orthogonal series estimation. We study theoretical properties of the proposed estimator establishing recovery of the latent nodal positions up to negligible error, 
and convergence of the multi-graphon estimator to the normal distribution.
Finite sample performance are investigated in a simulation study and application to
 two real-world networks---a dynamic contact network of ants and a collection of structural brain networks from different subjects---illustrate the utility of our approach. 
\end{abstract}

\section{Introduction}
Network data is commonly observed in a variety of real-world applications ranging from social networks observing interactions between pairs of individuals to biological networks such as protein-protein interactions.
This has led to a growing interest on probabilistic models for network data which not only offer a generative mechanism capturing empirically observed network effects, but are also easily estimable using existing statistical approaches. We study the setting where multiple distinct networks on the same set of nodes are observed. These may correspond to a collection of networks over an ordered set such as time (e.g., dynamic networks), or an unordered set such as networks from different subjects observed at a fixed point in time (cross-sectional networks). 
Given such datasets it is natural to ask:  how does the structure in networks change or evolve within the collection?

A nonparametric approach to modeling undirected network data is achieved by the graphon model~\cite{aldous85exch,diaconis2008graph,Bollobas11,Borgs17,lovasz12}, estimation of which has received a lot of attention (e.g.~\cite{Yang14,olhede14,wolfe13,Cai14,Veitch16,Pensky16, klopp2016, Klopp17, Padilla18}). However, this has mostly focused on estimation in the setting where only a single network is observed. Nonparametric modeling and estimation for multiple networks, in general assumed to be non-identically distributed, has largely been ignored. 
In many applications observing multiple networks, estimation under the general assumption of non-identically distributed networks seems natural. For example, consider a network of individuals with edges determined by similarity in political views,
observed at multiple time points.
Then in addition to a baseline model where political views are determined by a signature specific to each individual, a second source of variability arises from change in opinions over time as new information becomes available.
Without incorporating the second source of time-specific variability, we would average out important 
features which possibly characterise and differentiate interaction behavior at different time points. 
With this view, we propose a natural extension of the standard graphon model to incorporate network heterogeneity in addition to nodal heterogeneity via a multi-graphon function. Further, we show how information from multiple networks on the same set of nodes can be leveraged to enable estimation via standard nonparametric regression techniques for both replicated (i.i.d networks) and heterogeneous (independent but non-identically distributed) collection of networks.

The data consists of a collection of $m$ distinct undirected networks without self-loops, on the same set of $n$ nodes, represented using adjacency matrices $A_{1},\ldots,A_{m}$. These networks may be binary with each $A_{l}\in \{0,1\}^{n \times n}$, and $A_{ ijl}=1=A_{jil}, i\neq j$ indicating the presence of an edge between nodes $i$ and $j$ in the $l$th network; or weighted with $A_{ ijl}=A_{jil}$ recording the count of interactions between nodes $i$ and $j$ in the $l$th network. 
Given a single undirected binary network $G$, it is standard to assume that for $i\leq j$, $G_{ij}$ are independent Bernoulli$(P_{ij})$ trials, where $P_{ij}$ are edge probabilities determined by an underlying two dimensional function $f$, known as the graphon, e.g.~\cite{Zhang15nbd}. For a non-identically distributed collection of networks, $A_{1},\ldots,A_{m},$ 
we consider a natural extension of this model where $A_{ ijl}, i\leq j$ are independent Bernoulli$(P_{ijl})$ trials, with $P_{ijl}$ denoting edge-probability for node pair $(i,j)$ in the $l$th network. We achieve this via a three-dimensional bounded measurable function $f:[0,1]^{3}\rightarrow [0,1]$, we called \textit{multi-graphon}, where the third dimension allows for network-specific effects via network positions $z_{1},\ldots,z_{m}$ and thus different interaction probabilities in different networks. Further, the multi-graphon function by design is such that averaging over network-specific effects brings us back to the standard graphon model for replicated networks i.e., $A_{ ijl}$ as independent Bernoulli$(P_{ij})$ trials where $P_{ij}$ is now determined by the 
flattened multi-graphon $\bar{f}$ where $\bar{f}(x,y)=\int_{[0,1]} f(x,y;z)dz$. Change in interactions over distinct networks may arise as a result of a series of small changes occurring between consecutively observed networks or as a result of `jumps' (for example with $f$ as a stepfunction in $z$). In this paper, we focus on estimation of the multi-graphon array $[f_{ijl}, (i,j,l)\in [n]\times[n]\times[m]]$ for heterogeneous networks assumed to be generated from smooth kernels $f$ and hence of a `slowly-varying' type.

A key challenge in graphon estimation using standard nonparametric regression is the latency of nodal positions corresponding to the observed response of pairwise interactions. This has led to a variety of contributions focusing on histogram approximations to the graphon function, and more specifically, graphon matrix estimation (e.g.~\cite{Zhang15nbd,Padilla18, Airoldi13,Chan14,Yang14,olhede14,wolfe13,Cai14,Veitch16,Pensky16,Klopp17}). One of the main objectives in these methods is suitable identification of neighborhoods, either combinatorially (e.g.~\cite{olhede14,wolfe13}); through assumptions like strict monotonicity of the degree sequence~\cite{Chan14}; or through a construction of distances between node pairs (\cite{Airoldi13, Zhang15nbd, Padilla18}), each approach allowing a locally-averaged estimator. 
The method of~\cite{Zhang15nbd} is particularly attractive as it allows neighbors to vary from node to node
resulting in a local moving average estimator. Given the adaptive neighborhood choice, it is closer to a Nadaraya-Watson type estimator with uniform weighting in each neighborhood, than a standard histogram with fixed neighborhoods. 
While this offers a significant improvement over local-constant or histogram estimators, in general, it lacks the flexibility and advantages offered by the vast literature on standard nonparametric methods~\cite{wand94,fan18,smoothing2012} (with different smoothing techniques : local vs global, automatic smoothing parameter selection, direct implementation, to name a few).
Further, with the exception of~\cite{Airoldi13}, these methods are designed for graphon estimation from a single network. 
The method of~\cite{Airoldi13} provides a blockmodel approximation to graphon function using multiple i.i.d. networks and thus corresponds to the special case of replicated networks.

We propose a two step multi-graphon estimator where the first step 
uses the similarity of interactions between node pairs to construct an embedding of nodes in the Euclidean space, 
and the second step achieves estimation via nonparametric regression using estimated nodal positions from the first step as design points. 
Intuitively, embedding of nodes in the first step is based on the idea that for a smooth collection of networks, nodes `closer' to each other, must connect `similarly'. 
This leads to a concept of distances between pairs of nodes,
first studied in~\cite{Airoldi13} to cluster nodes into a fixed number of blocks, leading to a histogram approximation to graphon.
 Similar distance-based approaches have subsequently been used
for adaptive neighborhood selection~\cite{Zhang15nbd}, and more recently by~\cite{Padilla18} to allow estimation via fused lasso.
Unlike these existing approaches, we study the use of
pairwise distance comparisons of the form 
$\mathrm{dist}(i,j)<\mathrm{dist}(p,q), \forall \{i,j,p,q\}\in \{1,\ldots,n\}$
to identify nodal positions $\hat{x}_{1}, \hat{x}_{2},\ldots,\hat{x}_n
\in (0,1)$
via the ordinal embedding approach of~\cite{terada14}.  
Using classical Fr\'{e}chet bounds, we show that our pairwise nodal distance estimates concentrate jointly at exponential rates. Further using results related to the ``broken stick'' theorem, we prove that our maximal error in latent position estimate is $O(\log n / n)$. Leveraging this result, we find that our proposed method achieves, in a range of data sampling regimes --- this in terms of number of network observations, number of nodes they contain, and average network density --- the optimal convergence rate of an oracle estimator that observes the true latent positions.

In the special case of replicated networks arising from a common distribution, we are concerned with estimation of the standard two-dimensional graphon model and hence nonparametric regression
 is achieved easily using the estimated nodal positions. In the case of heterogeneous networks observed over time, it is assumed that network positions correspond to equi-spaced time points i.e., $z_{l}=t_{l}$, 
where $t_{l}=l/m, l\in [m]$, and our model reduces to the dynamic graphon model of~\cite{Pensky16}.
For heterogeneous cross-sectional networks, estimation of multi-graphon relies on the availability of network-level covariates which are modeled as noisy measurements of unobserved network positions.
Intuitively, this is motivated from the empirical observation that networks with similar traits (such as age or creativity scores of subjects in brain networks) often interact in ways similar to each other~\cite{Arroyo17,Vogelstein13}, and following related work such as~\cite{Fosdick15} modeling dependence between node covariates and unobserved node positions; covariates to explain link homophily~\cite{Yan19}. 

Finite sample performance studied via Monte Carlo simulations demonstrate that our method is comparable to existing methods for the case of replicated networks but performs significantly better when heterogeneous collection of networks are observed.
Useful insights on the performance of our two-step approach are offered by comparisons with oracle versions of our estimator obtained using knowledge of the true node and network positions. This offers a benchmark for comparison 
under a fixed choice of smoothing technique.
Further, we find that even with moderately informative network-level covariates (signal-to-noise ratio of one), the proposed estimator leads to significant improvements over existing methods in most cases.

We illustrate the usefulness of our approach using two real-world data sets: a contact network of ants observed over a period of $41$ days~\cite{Mersch13}, and human connectome networks from multiple subjects~\cite{Roncal2013,Kiar2016ndmg}. Our results reveal interesting insights on the division of labor among ant workers over time and on the link between brain region interactions and creativity levels. 
The multi-graphon model leads to newer insights which are lost when estimation is performed under the simplified assumption of replicated networks. 
Our multi-graphon estimates for the dynamic ant contact network suggest that changes in intensity of interaction between ant workers over time is possibly linked to changes in occupation of ant workers as they age (e.g., with younger nurse ants becoming cleaners over time). 
Multi-graphon estimates for the connectome networks revealed that intensities of interactions between certain brain region pairs may significantly increase and subsequently decrease (or vice versa) with increase in creativity scores, suggesting that high level analyses achieved via clustering of brain networks into low and high creativity groups (e.g.~\cite{DD18}), must be fine tuned to achieve a more accurate account of changes in brain region interactions with increase in creativity levels. 
An application of the estimated multi-graphon model to resampling brain networks shows that our estimated model captures the well-known small-world behavior of high creativity brains.

\section{Model Elicitation}
A probabilistic generative mechanism for a collection of $m$ heterogenous undirected networks, each on $n$ nodes, represented via adjacencies $A_{1},\ldots,A_{m}$, where each $A_{l}\in \{0,1\}^{n\times n}, l=1,\ldots,m$ is elicited via a multi-graphon defined below.

\begin{Definition}[Multi-graphon]
\label{meta-def}
We call multi-graphon a function $f: [0,1]^3\to[0,1]$,
such that for any given $z\in[0,1]$, $f(x, y; z)$ is a graphon in the conventional sense, i.e., integrable and $f(x, y; z)=f(y, x; z)$.
\end{Definition}

\begin{Definition}[Generalized random graph model $G(n,m,\rho_n f)$]
\label{model-def}
Let $(x_{1},\ldots,x_n)$ be a random vector sampled from a distribution $\mathbb{P}_{x}$ supported on $[0,1]^{n}$. Further, let $(z_{1},\ldots,z_{m})$ denote a random vector sampled from a distribution $\mathbb{P}_{z}$ supported on $[0,1]^{m}$. Given a multi-graphon $f$, conditional on the sampled positions $x_{i},x_{j},z_{l}$, we model $A_{ ijl}\in\{0,1\}$
for all $\{i,j\}\subset[n] \times [n]$, $l\in[m]$, as independent Bernoulli trials with
\[A_{ ijl}\,\vert\, x_i, x_j, z_l\sim\mathrm{Bernoulli}\big(P_{ijl}\big),\]
where $P_{ ijl}=\rho_n f(x_i, x_j; z_l)$, and 
$\rho_n\in (0,1)$, a decreasing function of $n$, determines the global sparsity of networks (e.g.~\cite{bollobas2007phase,bollobas2009metric,olhede14}).
\end{Definition}
For identifiability of $\rho_{n}$, it is assumed that $\int_{[0,1]^2}f(u, v; z)du dv=1$ for any $z\in[0,1]$. 
Then, clearly, for a binary network $\mathbb{E}(A_{ijl})=P(A_{ ijl}=1)=\rho_n$, and $\rho_n$ may be estimated as the average proportion of non-zero edges in each network, i.e.,
\[\hat{\rho}_n=\frac{\sum_{l=1}^{m}\sum_{i\leq j}A_{ ijl}}{m\binom{n}{2}}.\]
 
A significant proportion of the literature on dynamic (or multi-graph) network models are extensions of single network models augmented with a Markovian assumption to describe network evolution over time~\cite{Kolar2010,Matias2017}. Other related work includes latent space approaches modeling node and network dynamics through a single latent variable~\cite{Gollini16,Sewell2015,Xu2014,sarkar2006}.
On the other hand, our model assumes a common latent nodal space and a separate network-specific latent variable which allows varying interaction probabilities across network samples for any given pair of nodes. This feature allows a simple but flexible approach to capturing network-specific effects in a collection of slowly-varying networks, and has been studied in the context of multi-graph SBM~\cite{han2015,holland83} and more recently~\cite{Arroyo2019}.

\begin{Remark}
Following the Bernoulli model for $A_{ ijl}$ given above, weighted edges between nodes $i$
and $j$ in the $l$th network are conditionally independent Binomial random variables with success probabilities 
$P_{ ijl}$.
\end{Remark}

\begin{Definition}[Flattened $f$]
For a multi-graphon $f$, the flattened multi-graphon denoted as $\bar f$, is such that $\bar f(u,v)\mapsto\int_{[0,1]}f(u, v; z)dz$. 
\end{Definition}
Note that the flattened multi-graphon $\bar f$ is a graphon function. Following the literature on graphons and graphon estimation (e.g.~\cite{klopp2016,olhede14}), we assume $\mathbb{P}_{x}$ and $\mathbb{P}_{z}$ to be i.i.d. uniform denoted as $U[0,1].$ 

\section{Latent position estimation via embedding}\label{sec:OE}
The main goal of this section is to show how latent nodal positions can be inferred consistently using
a pairwise distance measure together with the ordinal embedding approach of~\cite{terada14}.
 We begin with the construction of a distance between pairs of nodes under the generalized random graph model with a smooth multi-graphon $f$. 
Subsequently, in~\cref{prop:consistency} we show that this distance can be estimated consistently from adjacencies $A_{1},\ldots,A_{m}\sim G(n,m,\rho_nf)$.
Further, we note that this distance, a semi-metric, corresponds to a metric on the purified graphon space. An important consequence of this fact is that nodal positions (or neighborhoods, e.g.~\cite{Airoldi13}) obtained via this distance correspond to positions of nodes in the purified graphon space.

\subsection{Distance between node pairs}
The concept of a distance between nodes of a network follows naturally for smooth multi-graphons: for node pairs $(i,j)$ closer to each other i.e., if $x_i$ is close to $x_j$, then for most $v$ and $z$, $f(x_i, v, z)$ and $f(x_j, v, z)$ should also be close (e.g.~\cite{Airoldi13}).
With this idea, the $l_{2}$ distance between multi-graphon planes at $x_{i}$ and $x_{j}$ may be used to quantify distance between nodes $i$ and $j$ as
\begin{equation}\label{distast-def}
\mathrm{dist}_{ij}(f)=\int_{[0,1]^2}\big(f(x_{i}, v; z)-f(x_{j},v; z)\big)^2dvdz.
\end{equation}
However, as we want to focus on the distance between vertices, which under the generalized random graph model (see~\cref{model-def}) can be recovered through the flattened graphon $\bar{f}$, it is sufficient to consider the distance based on the flattened graphon $\bar{f}$, i.e.,
\begin{equation}\label{def:dist}
\mathrm{dist}_{ij}(\bar{f})=\int_{[0,1]}\big(\bar{f}(x_{i}, v)-\bar{f}(x_{j},v)\big)^2dv.
\end{equation}

This distance can be estimated exactly using the adjacencies $A_{1},\ldots,A_{m}$ alone 
via~\cref{alg:dist} given below, which is a generalization of the algorithm in~\cite{Airoldi13} (see~\cref{rmrk:sparsity}), to allow robust estimation for networks, which may not necessarily be dense.

\begin{algorithm}
  \SetAlgoLined
\KwInput{A collection of $n\times n$ adjacencies $A_{1},\ldots,A_{m}$}
  \KwOutput{An $n \times n$ matrix $[\widehat{\mathrm{dist}}_{ij}(A)]_{i,j}$ measuring distances between node pairs}
  {\vspace{.2\baselineskip}
  With $S$ any $(\lfloor {m}/{2}\rfloor)$-subset of $[m]$, set $\hat r\in\mathbb{R}^{n\times n}$ such that $\forall i,j\in[n]$,
  $\hat r_{ij} =\tfrac{1}{n-2}\sum_{k\in[n]\setminus\{i,j\}}\left(\tfrac{1}{|S|}\sum_{l\in S} A_{ikl}\right)\left(\tfrac{1}{m-|S|}\sum_{l\in [m]\setminus S} A_{kjl}\right)$\;}
{Set $\widehat{\mathrm{dist}}_{ij}(A) = (\hat r_{ii}+\hat r_{jj}-\hat r_{ij}-\hat r_{ji})_+/\hat\rho_n^2$ $\forall i,j\in[n]$\;}
\caption{\label{alg:dist} Matrix distance estimator.}
\end{algorithm}
\begin{Proposition}[Consistency]\label{prop:consistency}
For $\{i,j\}\subset[n]$, if $(A,\cdot)\sim G(n,m,\rho_n f) \,\vert\, x_i,x_j$ and $\epsilon := (\rho_n^2 m^2n)^{-1}=o(1)$, then using~\cref{alg:dist}, and asymptotically in $n$ and $m$,
\[
\widehat{\mathrm{dist}}_{ij}(A) =\mathrm{dist}_{ij}(\bar{f})+\iota_{ij},
\]
where $\E\iota_{ij} = 0$ and $\iota_{ij} = O_p(\sqrt\epsilon)$.
\end{Proposition}
\subsubsection{Sparsity}\label{rmrk:sparsity}
From~\cref{prop:consistency} it is evident that we must have $\rho_n^2 m^2 n\to\infty$ for $\widehat{\mathrm{dist}}(A)$ to be consistent. For $m$ slowly increasing, this requires $n\rho_n = \omega(\sqrt{n})$; i.e., the average degree growing at least as fast as $\sqrt{n}$. Put another way, 
it requires the number of paths of length two between any two nodes to behave like a $\mathrm{Poisson}(\rho_n^2m^2n)$, and we need the mean $\rho_n^2m^2n$ to be large enough to carry a Normal approximation.
It follows that we are assuming that the total number of paths of length two between any pair of nodes across network replicates is in general larger than 20. This assumption could be unrealistic for some sparse networks . In case the assumption cannot be met we suggest the following modification to~\cref{alg:dist}: instead of counting paths of length $2$ between nodes $i$ and $j$ to define $\hat r$ in Step 1., use paths of length $2e$, for integer $e>1$; i.e., set
\[
\hat r_{ij}^{(e)} = \frac{1}{n-2}\sum_{k\in[n]\setminus\{i,j\}}\left(\frac{1}{|S|}\sum_{l\in S} A_{ikl}^e\right)\left(\frac{1}{m-|S|}\sum_{l\in [m]\setminus S} A_{kjl}^e\right).
\]
Then, it is possible to first show with a direct walk counting argument that, (see e.g.~\cite{BickelLevina2012})
\[
A_{ikl}^e = \#\{\text{paths of length $e$ between $i$ and $k$ in $A_{\cdot\cdot l}$}\}+O_p\big((n\rho_n)^{-1}\big).
\]
Then, by the exact same steps as in the proof of~\cref{prop:consistency}, we obtain that with ${\bar f}^{(e)}(u,v) = \int_{[0,1]^e}f(u,y_1;z)f(y_1,y_2;z)\cdots f(y_{e-1},v;z)dy_1dy_2\cdots dy_{e-1}dz$,
\[
\widehat{\mathrm{dist}}_{ij}^{(e)}(A)
  =\int_{[0,1]}\big({\bar f}^{(e)}(x_i, z) -
        {\bar f}^{(e)}(x_j, z)\big)^2dz
    + O_p\big((n\rho_n)^{-1}+(m^2\rho_n^{2l}n^{2e-1})^{-1}\big).
\]
This reduces the density requirement to, $n\rho_n=\omega(\sqrt[2e]{n}),$ for $m$ finite, at the cost of a coarser distance.
There is also a computational cost. Indeed, while both the space and computational complexity of~\cref{alg:dist} are $O(n^2m)$, the modified version above has the same space complexity, but computation are $O(n^\varsigma m)$ (with $\varsigma$ the complexity of the matrix product.)
\subsubsection{Pure graphons}\label{rmrk:pure}
A characterization of the distance given by \ref{def:dist} follows through its association with a metric induced by $\bar f$. 
With $D$ a distribution of latent nodal positions on $[0,1]$ and $f$ a graphon, let
\begin{equation*}\label{dist-def-2}
\mathrm{dist}\big((\bar f, D); x_i, x_j\big) := \E_{u\sim D}\left[\big(\bar f(x_i, u)-\bar f(u, x_j)\big)^2\right],
\end{equation*}
Then, $\mathrm{dist}((\bar f, D); \cdot, \cdot)$ is a semi-metric on $(0,1)$~\cite[Section 13]{lovasz12}. For example, in our case, noting that $\bar f: (0,1)^2\to(0,1)$ 
is a positive symmetric operator, and assuming $\bar f$ to be of finite rank $r>0$, we may write $\bar f(u,v)=\sum_{p\leq r} \lambda_p \varphi_p(u) \varphi_p(v)$ where the $\varphi_p$ form an orthogonal basis; i.e., for all $p, q$, $\int\varphi_p(u)\varphi_q(u)du = \bm1\{p=q\}$.
Then, writing $\varphi(u) = \big(\sqrt{\lambda_p}\varphi_p(u)\big)_{p\in[r]}$, and setting $D=\mathrm{U}[0,1]$ the uniform distribution on $[0, 1]$, we observe that
\begin{equation*}
\mathrm{dist}\big((\bar f, D); u, v\big) = \|\varphi(u)-\varphi(v)\|_2^2,
\end{equation*}
thereby proving that $\mathrm{dist}((\bar f, \mathrm{U}); u, v)$ is the Euclidean distance between the images of $u$ and $v$ projected by $\varphi$. 
However, by~\cite[Subsection 13.3]{lovasz12}, $\mathrm{dist}((\bar f, D); \cdot, \cdot)$ can be transformed into a metric via purification of $\bar{f}$. 
Specifically, for a graphon $\bar{f}$, there exist maps $\psi: [0,1]\to J$ and $\bar f^\ast: J^2\to[0,1]$ such that:
\begin{enumerate}
\item $\bar f^\ast(\psi(u),\psi(v))=\bar f(u, v)$ almost everywhere for i.i.d. $u, v \sim D$, and
\item $\mathrm{dist}\big((\bar f^\ast,\psi(D)); \cdot, \cdot\big)$ is a metric on $J$,
\end{enumerate}
and $\bar f^\ast$ is referred to as the \emph{purified graphon} corresponding to $\bar{f}$. %
In~\cite[Section 13]{lovasz12}, arguments are presented motivating the assumption that graphons, except some pathological cases, can be purified in such a way that $J$ is of dimension one.

\subsection{Node embedding}
As discussed above, our goal is to 
obtain nodal positions satisfying distance comparisons implied by
$\widehat{\mathrm{dist}}(A)$. While we could, for instance use the Gram operator, the quality of the estimate would only scale, at best, with $\sqrt\epsilon$, as seen in~\cref{prop:consistency}. 
We note that this rate can be significantly improved through ordinal embedding~\cite{terada14,Arias2017}. 
To justify the use of ordinal embedding we must first show that our distance estimator will order the distances appropriately with high probability. 
This is the case in our setting, as shown below in~\cref{consistent-latent-ordering}.
We establish consistency of our nodal position estimator (up to a similarity transformation) in~\cref{latent-est-rate}.
\begin{Proposition}
\label{consistent-latent-ordering}
For $\{i,j,p,q\}\subset[n]$, if $(A,\cdot)\sim G(n, m, \rho f) \,\vert\, x_i, x_j, x_p, x_q$ and $\epsilon := ((\rho^2_n m^2n)^{-1}=o(1/\log n)$, then using~\cref{alg:dist} there exists $c>0$ such that for $n$ and $m$ large enough
\[
\P\left[\frac{
  \widehat{\mathrm{dist}}_{ij}(A) - \widehat{\mathrm{dist}}_{pq}(A)}
  {\E\left[\widehat{\mathrm{dist}}_{ij}(A)-\widehat{\mathrm{dist}}_{pq}(A)\right]}
  >0\right] \geq 1-e^{-c/\epsilon}.
\]
Then, by the Fr\'{e}chet inequality, asymptotically in $n$ and $m$,
\[
\P\left[\forall \{i,j,p,q\}\subset[n],\ \frac{\widehat{\mathrm{dist}}_{ij}(A)-\widehat{\mathrm{dist}}_{pq}(A)}{\E\left[\widehat{\mathrm{dist}}_{ij}(A)- \widehat{\mathrm{dist}}_{pq}(A)\right]}
  >0\right] \geq 1-n^4e^{-c/\epsilon}\to 1.
\]
\end{Proposition}
From the characterization of distance via the purified graphon, it follows that using ordinal embedding on $\widehat{\mathrm{dist}}(A)$ will yield an estimate of the the latent positions under the purified graphon (the $\psi(x_i)$s). Theorem \ref{latent-est-rate} shows that this estimate is consistent, up to similarity transform, with an error bounded by $\log n / n$.
\begin{thm}\label{latent-est-rate}
Ordinal embedding with $\widehat{\mathrm{dist}}(A)$ produces consistent (up to similarity transform) estimators of the latent vertex location under the purified graphon, with a maximal error of order $\log n / n$.
\end{thm}
Indeed, ordinal embedding positions converge at the same rate that the latent positions cover the latent space~\cite[Theorem. 3]{Arias2017}. Since the latent space is $(0,1)$ and the latent positions are i.i.d. $\mathrm{U}(0,1)$, we achieve a rate of $\log n / n$ by the broken stick theorem (see details in~\cref{app:Proofs}).

\section{Multi-graphon estimation}\label{sec:algf}
The algorithm for multi-graphon estimation based on embedding nodal positions 
 is included below. 
\Cref{main-clt} shows that the resulting multi-graphon estimator is consistent for a family of piecewise Lipschitz graphon functions.
Further, our estimator achieves the optimal rate of $\sqrt{n+m}$, as if the latent positions were observed.
Given an $n \times n$ matrix $G$, let $\text{vec}\{G\}$ denote vectorization of $G$ into an $n^2$ length column vector obtained by stacking the transposed rows of $G$, on top of one another. 
\begin{algorithm}\label{alg:estim}
  \SetAlgoLined
  \KwInput{Adjacency matrices $A_{1},\ldots,A_{m}$, each $n \times n$, observed at time points $t_{1},\ldots,t_{m}$ (dynamic networks), or with
 network-level covariates $\check{z}_{1},\ldots, \check{z}_{m}$, each $\check{z}_{l}\in [0,1]$ (for cross-sectional networks)}

 \KwOutput{$\{\hat{f}_{ ijl}; (i,j,l) \in [n] \times [n] \times [m]\}$}
  {\vspace{.2\baselineskip}
  Use Algorithm \ref{alg:dist}. to construct $\widehat{\mathrm{dist}}(A)\in \mathbb{R}^{n \times n}$\;
 {Use $\widehat{\mathrm{dist}}(A)$ to obtain nodal position estimates $\hat{x}_{1},\ldots,\hat{x}_n$ via ordinal embedding~\cite{terada14}\;}
{Perform smoothing via standard approaches such as kernel regression, regression splines, 
to estimate $\hat{P}_{ ijl}=g(\mathbb{E}(y|\hat{x}_{i}, \hat{x}_{j}, \tilde{z}_{l}))$ with $y=[\mathrm{vec}\{A_{ ijl}\}]_{(i,j,l)}$ as the $n^2m$ length response vector corresponding to node-network positions $[\hat{x}_{i}, \hat{x}_{j}, \tilde{z}_{l}]_{i,j,l}, (i,j,l) \in [n] \times [n]\times[m]$, where $\tilde{z}_{l}=t_{l}$ for dynamic networks and $\tilde{z}_{l}=\check{z}_{l}$ for cross-sectional networks, and $g$ denotes a link function (e.g. logit for binary networks)\;}
{Set $\hat{f}_{ ijl}=\hat{\rho}_n^{-1}\hat{P}_{ ijl}$, $i,j \in [n], l\in [m]$ \;}
\caption{\label{alg:meta} Multi-graphon estimator.}}
\end{algorithm}

\begin{thm}
\label{main-clt}
Fix a smooth multi-graphon function $f$. Assume that we observe $\check z_l$, noisy measurements of the true network positions $z_l$, such that $\check z_l - z_l$ has finite second moments. Call $\mathcal{D}$ the joint distribution of a pair of latent $x_i$'s and $\check z_k$. Set $h:[0,1]^4\to\mathbb{R}$ such that $h$ is symmetric in its first two arguments, linear in the fourth, and that $h$ and its first derivatives are finite almost everywhere. Then, if $\epsilon := (\rho_n^2 m^2n)^{-1}=o(1/\log n)$ and $m=o\big((n/\log n)^2\big)$, asymptotically in $n$ and $m$,
\begin{equation}\label{main-clt-eq}
\sqrt{n+m}\Bigg(
  \frac{1}{n^2m}\sum_{i, j, l}h\big(\hat x_i, \hat x_j, \check z_l, A_{ ijl}\big)-
  \E_{(u, v, s)\sim \mathcal{D}} h\big(u, v, s, f(u, v; s)\big)
\Bigg)\to\mathrm{Normal}\big(0,\Sigma\big).
\end{equation}
\end{thm}

\cref{main-clt} shows that in the setting we consider (in effect, independent observations from a smooth multi-graphon with $(\rho_n^2 n)^{-1/2}\ll m\ll n^2$), estimation of the latent nodal positions comes at negligible accuracy cost. Indeed, the rate of convergence we obtain is $\sqrt{n+m}$, which is the same rate as the optimal rate we could obtain if the latent position were observed~\citep{Grams73}. This naturally raises the question of what concretely this regime encompasses, and its limits.

The first case to consider is when $m$ remains small, which corresponds most closely to the setting where only a single adjacency matrix is observed. Then, our assumption translates into an assumption on the density of the network --- specifically $\rho_n \gg 1/\sqrt{n}$ --- which will be unrealistic in some settings; e.g., social network observations tend to be much sparser in practice, with $\rho_n$ in the range of $1/n$ to $\log n / n$~\cite{Barabasi99}. However, other applications, such as connectome networks could accommodate such a density regime~\citep{Maugis17}. This point puts into perspective~\cref{rmrk:sparsity}, which allows to relax the assumption for~\cref{main-clt} in this setting to $\rho_n \gg 1/\sqrt[k]{n}$ for any $k$, at a computational and bias cost.

Next, consider the case where network density $\rho_n$ is in the range of $1/n$ to $\log n / n$, as has been observed in many settings~\cite{Barabasi99}. Then our assumption translates to $m\gg n$, which is demanding, especially when $n$ is large. Here we note that while $m\gg n$ is indeed demanding, it is not unreasonable in the sense that~\cref{main-clt} provides local graph statistics, specifically point-wise estimate of all edges probabilities, and it could easily incorporate node specific covariates. If the goal of estimation was instead to evaluate global estimates, say averaged across nodes or edges such as motif counts~\cite{Maugis17}, then the assumption could be relaxed.

Based on the results and remarks included above, we provide the recommended estimation approach when $(n,m,\rho_n)$ lie outside the regime of~\cref{main-clt}:
\begin{enumerate}
\item If $m \gg n^2$, then one should perform $n^2$ regressions, one for each pair of vertices, where the response variable are the observed edges between the selected vertices. Thus, for each fixed node pair $(p,q)\in [n]\times [n]$, $y_{pq}=[\mathrm{vec}\{A_{pql}\}]_{l}$ as the length $m$ response vector corresponding to 
$[\tilde{z}_{l}]_{l}, l\in [m]$. The achieved rate will match ours in that regime, namely $\sqrt{m}$, but will be much lighter computationally, and fully parallelizable. Intuitively, the idea is to borrow information from `neighboring' networks (in time or with similar traits) rather than neighboring nodes due to $m$ being much larger than $n^2.$
\item If $m \ll (n\rho_n^2)^{-1/2}$, then one should estimate a graphon $\hat{f}$ for each observed network separately (using an existing approach for single networks, e.g.~\cite{olhede14,Zhang15nbd}), and subsequently perform $n^2$ local regressions, 
one for each pair of vertices with the estimated edge intensities as response i.e., $y_{pq}=[\mathrm{vec}\{\hat{f}_{pql}\}]_{l}$ and network level covariates $[\tilde{z}_{l}]_{l}, l\in [m]$ as regressors.

This follows from~\cref{rmrk:sparsity}, and the achieved rate will depend on the smoothness of the multi-graphon, but the said rate will be affected by the sparsity $\rho_n$; e.g., a graphon estimate with $\sqrt{n}$ blocks (a standard choice for number of blocks~\cite{olhede14}), will converge at most at rate $\sqrt{n\rho_n^2}$~\citep{wolfe13}, much slower than the rate under the $(n,m,\rho_n)$ regime of~\cref{main-clt}.
Note that~\cref{rmrk:sparsity} allows for~\cref{main-clt} to apply to cases where $m \gg n^k\rho_n^{k+1}$ for some $k$.
\end{enumerate}
Therefore, we conclude that~\cref{main-clt} yields optimal rates for local graph statistics in the regimes it applies to.

\begin{Remark}
In the special case of replicated or i.i.d networks, we are concerned with estimation of a common network generating process or the standard two-dimensional graphon $[\bar{f}(x_{i},x_{j}); (i,j)\in [n]\times [n]]$.
Using the aggregated adjacency $\bar{A}=\sum_{l=1}^{m}A_{...}/m$ and the estimated nodal positions as above, we arrive at a special case of~\cref{main-clt} given by~\cref{main-clt-rep} in~\cref{app:Proofs}, which shows that local regression with $y=[\mathrm{vec}\{\bar{A}_{ij}\}]_{i,j}$ as the length $n^2$ response vector with estimated nodal positions $[\hat{x}_{i},\hat{x}_{j}]_{i,j}$ as the regressors, leads to a graphon estimator which enjoys the same properties as the multi-graphon estimator.
\end{Remark}

Further, our algorithm for multi-graphon estimation with kernel regression using a uniform kernel in Step 3. may be viewed as an extension to the neighborhood smoothing approach of~\cite{Zhang15nbd} (designed for single networks) to the setting of multiple networks, with neighborhood identification based on ordinal embedding. In general, our approach has the key advantage of enabling standard nonparametric regression techniques due to the availability of nodal position estimates.

\section{Finite sample performance}\label{sec:sim}
We conducted simulations to study finite sample performance of the proposed two-step multi-graphon estimator
for a synthetic collection of $m$ networks, each on $n$ nodes, generated using functions $f$ with different degrees of smoothness 
and in general, with network-specific variability.
Consider the following three multi-graphon functions:
\begin{enumerate}
\item $f_{1}(x,y;z)=(xy+\beta z^2)/(0.25+\beta z^2)$
\item $f_{2}(x,y;z)= (\exp(-|x-y|/2)+\beta z)/(0.8522+\beta z)$
\item $f_{3}(x,y;z)=r_{az}\mathbb{I}_{a=b}+ r_{abz} \mathbb{I}_{a\neq b}$, where $a=\lceil{kx}\rceil$ and $b=\lceil{ky}\rceil$, $k=2$ (number of blocks), and $r_{az}=0.7-0.0938{\beta z}$, $r_{abz}=0.3+\beta xyz$,
\end{enumerate}
where in each example,
setting $\beta>0$ allows for heterogeneity across network samples $A_{1}, \ldots, A_{m}$ through the network specific positions $z_{1}, \ldots,z_{m}$, whereas $\beta=0$ implies a replicated network sample where $A_{l}$, for each $l\in [m]$ arises from a
common distribution specified by $f(x,y)$. 
Given $f_{j}, j=1,2,3$, heterogeneous networks were generated using $\beta=0.35, 0.5, 0.6$, respectively.
Intuitively, our choice of $\beta$'s is such that it prevents the extremely smooth product kernel ($f_{1}$) from approaching a constant with increasing $z$, and on the other hand, allows the discrete-blockmodel ($f_{3}$) to gain some smoothness across blocks with increasing $z$. 
The structures implied by multi-graphons $f_{1}, f_{2}, f_{3}$ with increasing network positions,
precisely, (a) $z=0.05$, (b) $z=0.5,$ and (c) $z=0.95$, are displayed in~\cref{fig:grs}.

The first multi-graphon $f_{1}$ determines links between pairs of nodes based on the product of node-specific factors $(x,y)$ and with additive network-specific effects via $z$, implying a smooth surface.
 The smooth structure of $f_{1}$ appears ideal for nonparametric regression, however, 
this may also lead to a high variance in nodal position estimates 
 due to similar distances between subsets of nodes.
This example is designed to understand the trade-off between these two aspects.
The second graphon $f_{2}$ has a \textit{Robinsonian} form (e.g., Hubert et al. (1998)) with a peak on the diagonal and decreasing intensity as one moves away from the diagonal on either side. 
The third graphon $f_{3}$ is a simple stochastic blockmodel with $k=2$ blocks in the case of replicated networks i.e., $\beta=0$. 
Clearly, the probability of interaction between nodes across blocks is determined via $r_{abz}$ with the network-specific factor $z$ interacting with node-specific positions $(x,y)$. Thus, across-block probabilities increase non-uniformly across nodes, whereas, within-block probabilities determined via $r_{az}$ (no interaction term) decrease uniformly across all nodes within the two blocks.

Given $(n,m,\rho_nf)$, a generalized random graph sample comprising adjacencies $A_{1},\ldots,A_{m}$ is simulated via independent Bernoulli trials following~\cref{model-def}. We use uniformly distributed latent nodal and network positions i.e., $x_{1},\ldots,x_n\overset{\text{i.i.d}}{\sim} U(0,1)$, and $z_{1},\ldots,z_{m}\overset{\text{i.i.d}}{\sim} \mathrm{U}(0,1)$. Further, network-level covariates $\check{z}_{l}, l \in [m]$ are sampled as noisy measurements of the corresponding unobserved network-specific positions ${z}_{l}, l \in [m]$, i.e.,
\begin{equation}\label{eqn:covv}
\check{z}_{l}=z_{l}+ \epsilon_{l}, \text{  where  } \epsilon_{l} \sim N(0,\sigma^2).
\end{equation}
Clearly, the quality of network-specific covariates $\check{z}$ as measurements of the unobserved latent positions $z$ is a function of the noise variance $\sigma^2$. Since $z\sim \mathrm{U}(0,1)$ in our simulation set-up, we chose $\sigma=0.28$ implying a signal-to-noise ratio (SNR) of $\approx 1$. An SNR of unity implies that the `signal' (covariate) is only as strong as noise and thus allows us to examine the performance and robustness of our method in settings 
where the observed covariates may not be ideal measurements of the true latent network-specific positions. 

\begin{figure}[t]
\centering
\includegraphics[width=0.5\textwidth]{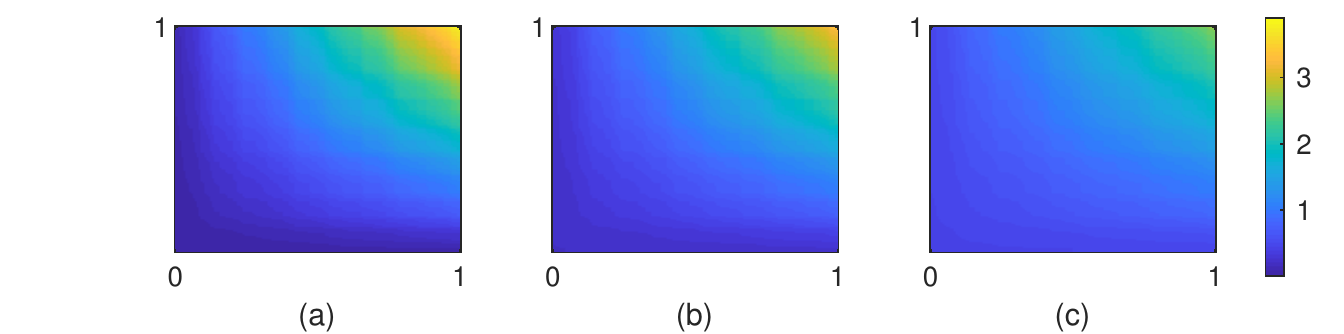}
\bigbreak
\includegraphics[width=0.5\textwidth]{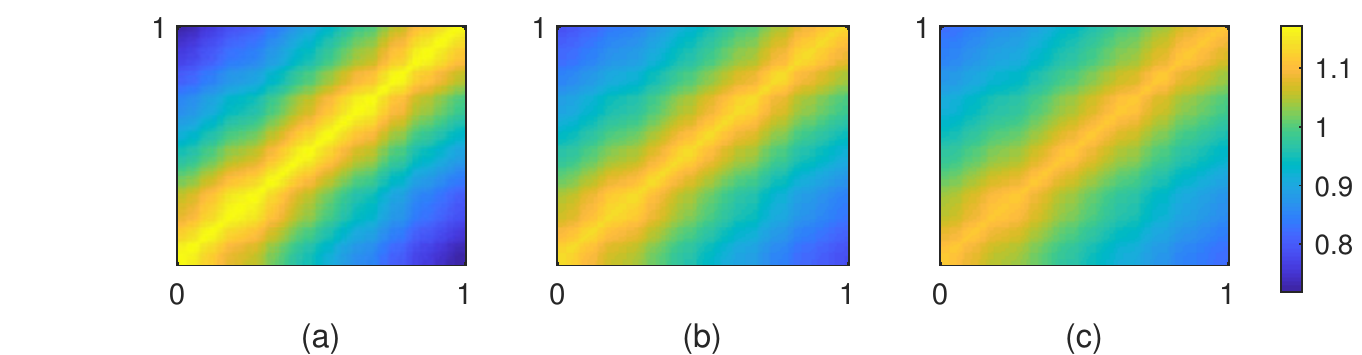}
\bigbreak
\includegraphics[width=0.5\textwidth]{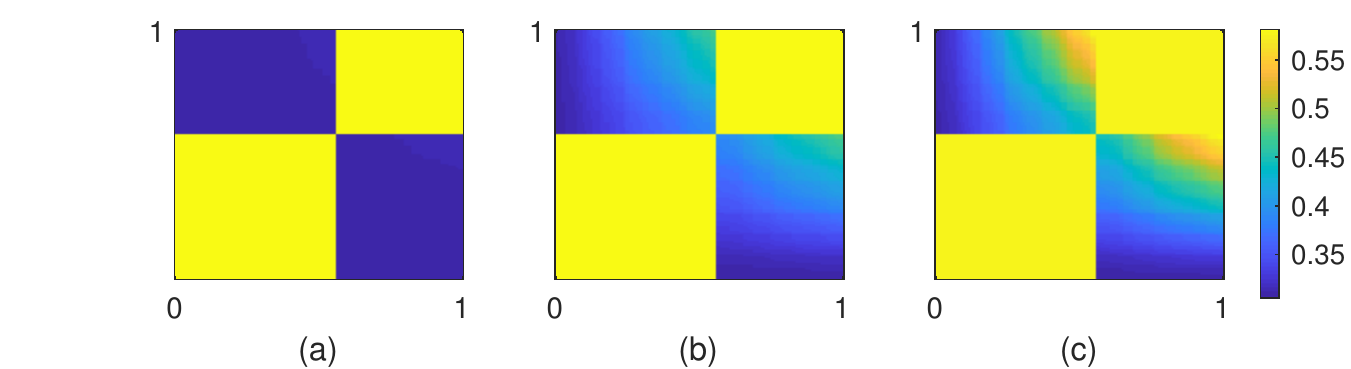}
\caption{Synthetic multi-graphon functions $f_{1}$ (top row), $f_{2}$ (middle row), and $f_{3}$ (bottom row), with increasing network positions $z$ across columns: (a) $f(.,.;z=0.05)$, (b) $f(.,.;z=0.5)$, and (c) $f(.,.;z=0.95)$.
Here $n=150$, $m=100$.}
\label{fig:grs}
\end{figure}

We compare the performance of our two step multi-graphon estimator with competing methods of SBA~\cite{Airoldi13}, SAS~\cite{Chan14}, USVT~\cite{Chatterjee15} and NBS~\cite{Zhang15nbd}. 
The algorithm of SBA achieves graphon estimation from a sample of multiple i.i.d. networks and hence corresponds to our case of replicated networks ($\beta=0$).
In order to compare with SAS, USVT and NBS, designed to work with a single adjacency matrix, we report results obtained using the aggregated adjacency $\bar{A}=\sum_{l=1}^{m}A_{..l}/m$. As far as we are aware, no competing methods exist for nonparametric estimation of the heterogeneous network generating process 
given a collection of independent, non-identically distributed networks.
Noting this, we report comparisons of estimates obtained with our approach under oracle settings 
described below.

The simulations are conducted with a view to understand the performance of our approach for a given choice of nonparametric regression in Step 3. of \cref{alg:meta}.
This may not always lead to the smallest possible MSE using our method but shall give us a view of the general finite sample performance.
We report results obtained with orthogonal series estimation with thin plate regression splines as basis functions~\cite{Wood2003}. This was implemented in {\tt{R}} using \textit{bam} in package \textbf{gam}. In general, our approach can be easily implemented in {\tt{R}} using other smoothing techniques e.g., the Nadaraya Watson estimator which may be implemented using \textit{kernreg} in package \textbf{gplm}.

\subsection{Replicated networks ($\beta=0$)}\label{sec:rep}
A comparison of our approach with existing methods based on MSE averaged over $50$ replications are reported in Table \ref{tab:simnc}; visual comparisons from a single run are displayed in~\cref{fig:grnc_es}.
To interpret performance of our two-step estimator, we 
consider an oracle setting where the oracle informs order-statistics of the true latent node-specific positions
(rather than the exact nodal positions). This information is used to directly construct oracle nodal position estimates denoted as $\hat{x}^{*}_{(i)}=i/(n+1)$~\cite{olhede14,Davison97}, using which nonparametric regression is performed following Step 3. of~\cref{alg:meta}.
We refer to this as the oracle graphon estimator. Note that our oracle set-up 
does not assume the nodal positions to be known and
 is designed to be closer to the actual set-up involving unobserved design points. 

First, comparing MSEs for estimates from 
the proposed method under the non-oracle setting (`Proposed') with the oracle setting (`Proposed$^{*}$'), we note significant differences between the two for $f_{1}$, and negligible difference for $f_{3},$ 
 across all sample sizes $(n,m), n>50$.
This indicates that the first step of latent position estimation performs poorly for $f_{1}$ and extremely well for $f_{3}$. This is what we expect due to the smooth structure of $f_{1}$ leading to subsets of nodes with similar distances and hence resulting in latent position estimates with high variance. The discrete structure of $f_{3},$ on the other hand, allows clearer separation between node pairs corresponding to the two blocks due to significantly different distances, implying robust latent position estimates (as far as blockmodel estimation is concerned). Similarly, comparing oracle and non-oracle MSEs for $f_{2}$ indicate that latent position estimation works reasonably well for these networks. 

In comparison to existing approaches, our actual proposed estimator (non-oracle) leads to the smallest MSE for $f_{1}$
in all cases except when $n=50$. A significant reduction in the MSE of $f_{1}$ is observed as $n$ is increased from $n=50$ to $n=100$, suggesting that $n=50$ nodes are insufficient to perform reliable estimation for $f_{1}.$ For $f_{2}$, our approach consistently leads to the smallest MSE with NBS leading to the second best performance. The relatively higher variance of estimates from our approach is due to high variance in nodal position estimation across replications. As discussed earlier, this is due to the smooth structure of $f_{2}$ (interestingly, heterogeneity across networks reduces the variance in nodal position estimates significantly for $f_{1}$ and $f_{2}$: see results reported in~\cref{sec:heter}). In practice, we recommend re-running the first ordinal embedding step a few times and subsequently selecting the nodal embedding with the lowest stress~\cite{terada14}, as this resulted in a reduced overall variance of estimates from the proposed method. For $f_{3}$, SBA, USVT and NBS lead to the best results with SAS leading to the highest MSE. The relatively higher MSEs from our approach for $f_{3}$ is due to the choice of nonparametric regression, precisely splines as basis functions which are clearly not ideal for estimation of a discrete blockmodel. This is evident from MSEs under the oracle setting, which are also high and comparable to MSEs under the actual non-oracle setting.

We observe that MSEs decrease with increase in $n$ for fixed $m=150$ in all cases, however, this is not necessarily the case with increase in $m$ and $n=150$ fixed, for $f_{1}$ and $f_{2}$. This appears to be an artefact of estimation being performed with a different number of adjacencies (precisely $m$) aggregated in each case, generated from functions with high degree of smoothness ($f_{1}$ and $f_{2}$).

\begin{table}
\caption{Mean squared error ($\pm$ std. dev.) comparisons of graphon estimates, all multiplied by $10^3$, averaged over $50$ replications. Proposed$^{*}$ (proposed under oracle), SBA of~\cite{Airoldi13}, SAS of~\cite{Chan14}, USVT of~\cite{Chatterjee15} and NBS of~\cite{Zhang15nbd}.}
 \label{tab:simnc}
\begin{center}
\resizebox {\textwidth }{!}{
 \begin{tabular}{l l l| l l l l l l }
Graphon & $n$ & $m$ & Proposed$^{*}$ & Proposed & SBA & SAS & USVT & NBS \\
 \hline
\multirow{ 6}{*}{$f_{1}$} &50 & 150 & $26.80 (23.60)$ & $92.00 (82.10)$ & $339.40 (185.10)$ & $83.60 (60.80)$ & $48.10 (55.20)$ & $54.50 (53.20)$ \\
      & 100 & 150 & $11.00 (9.70)$ & $17.20 (18.60)$ & $240.50 (118.20)$ & $63.30 (43.60)$ & $21.40 (28.70)$  & $26.00 (27.30)$ \\
    & 150 & 150 & $8.70 (7.60)$ & $14.30 (16.80)$ & $272.80 (215.20) $ & $26.00 (22.90) $ & $14.70 (20.30)$ & $15.60 (17.60) $\\
   & 150 & 50 & $10.10 (12.70)$ & $17.60 (30.30)$ & $181.90 (214.80)$ & $ 40.40 (46.40)$ & $21.70 (34.80)$ & $24.50 (35.00)$ \\
    & 150 & 100 & $6.20 (4.60)$ & $11.30 (13.70)$ & $186.10 (199.80)$ & $24.40 (17.80) $ & $10.90 (12.20)$ & $12.50 (10.60) $ \\
    & 150 & 150 & $8.70 (7.60)$ & $14.30 (16.80)$ & $272.80 (215.20) $ & $26.00 (22.90) $ & $14.70 (20.30)$ & $15.60 (17.60) $\\
	\hline\\
\multirow{ 6}{*}{$f_{2}$} & 50 & 150 & $2.00 (0.58)$ & $4.70 (3.90)$ & $7.90 (2.80)$ & $11.50 (1.50)$ & $10.50 (1.60)$ & $6.70 (0.83)$   \\
    &100 & 150 & $0.72 (0.23)$ & $1.60 (2.90)$ & $5.40 (3.00)$ & $11.10 (1.20)$ & $10.90 (1.30)$ & $2.70 (0.74)$  \\
    &150 & 150  & $0.43 (0.14)$ & $0.96 (2.20)$ & $6.00 (4.50)$ & $10.40 (0.87)$ & $10.00 (2.40)$ & $1.40 (0.44)$  \\
   &150 & 50 & $0.44 (0.16)$ & $1.00 (1.60)$ & $4.20 (3.20)$ & $10.20 (0.72)$ & $10.20 (1.50)$ &$1.60 (0.38)$  \\
   &  150 & 100 & $0.45(0.13)$ & $0.79 (1.30)$ & $3.00 (3.10)$  & $10.40 (1.00)$ & $9.70 (2.70)$ & $1.50 (0.53)$  \\
   & 150 & 150 & $0.43 (0.14)$ & $0.96 (2.20)$ & $6.00 (4.50)$ & $10.40 (0.87)$ & $10.00 (2.40)$ & $1.40 (0.44)$  \\
	\hline\\
 \multirow{ 6}{*}{$f_{3}$} & 50 & 150 & $9.70 (5.80)$ & $10.90 (7.50)$ & $2.70 (7.20)$ & $14.70 (13.80)$ & 0.86 (2.60) & $0.75 (0.08)$  \\
    &100 & 150 & $8.00 (4.60)$ & $8.30(5.30)$ & $0.25 (0.05)$ & $10.60 (10.70)$ & $0.15 (0.01)$ & $0.27 (0.01)$   \\
    &150 & 150 & $8.00 (3.10)$  & $7.80 (3.06)$ & $0.09 (0.02)$ & $9.70 (12.20)$ & $0.08 (0.005)$ & $0.02 (0.006)$  \\
    &150 & 50 & $7.60 (2.80)$ & $7.90 (3.30)$ & $0.15 (0.009)$  & $13.20 (13.20)$ & $0.15 (0.01)$ & $0.24 (0.01)$  \\
   &  150 & 100 & $8.00 (3.60)$ & $7.80 (3.20)$ & $0.16 (0.04)$ & $11.40 (12.50)$ & $0.10 (0.01)$  & $0.17 (0.01)$  \\
   &150 & 150 & $8.00 (3.10)$ & $7.80 (3.06)$ & $0.09 (0.02)$ & $9.70 (12.20)$ & $0.08 (0.005)$ & $0.02 (0.006)$  \\
\end{tabular}}
\end{center}
\end{table}

\begin{figure}
\centering
\includegraphics[width=0.57\textwidth]{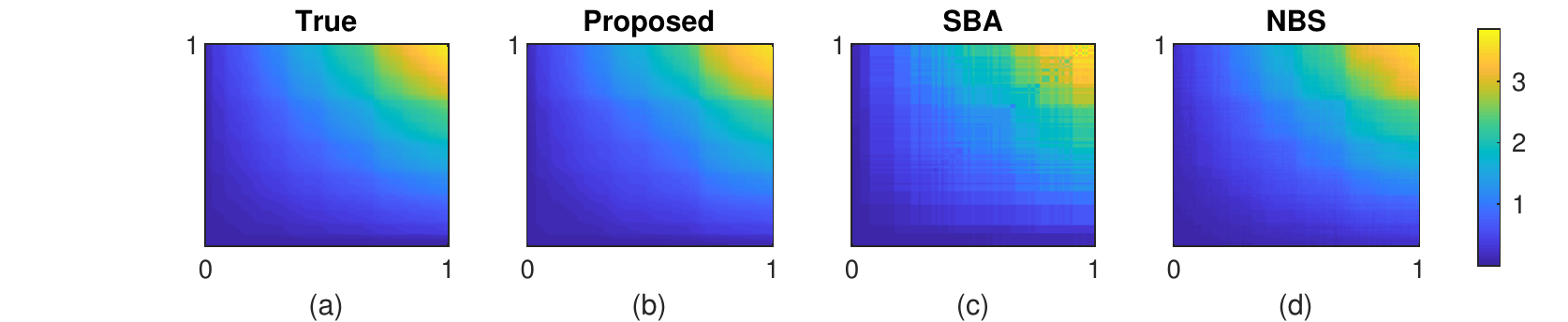}
\bigbreak
\includegraphics[width=.57\textwidth]{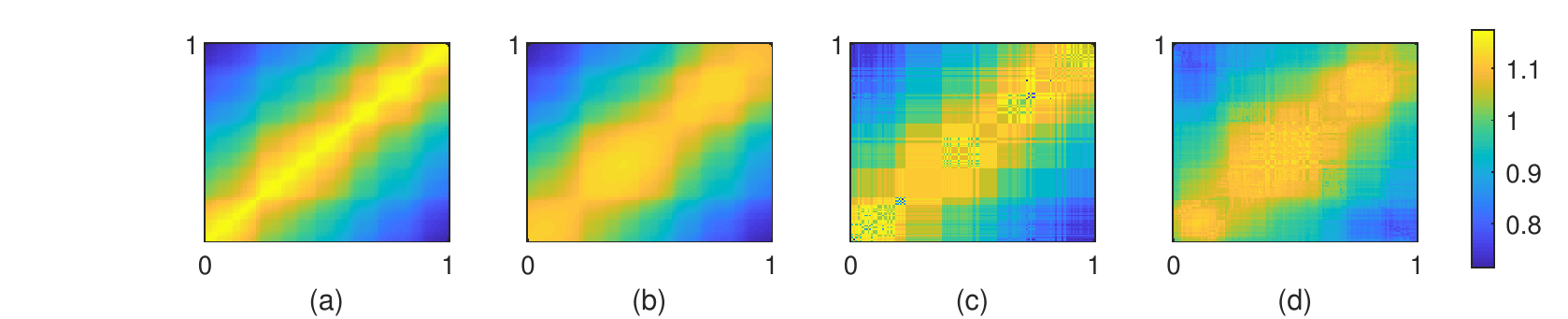}
\bigbreak
\includegraphics[width=.57\textwidth]{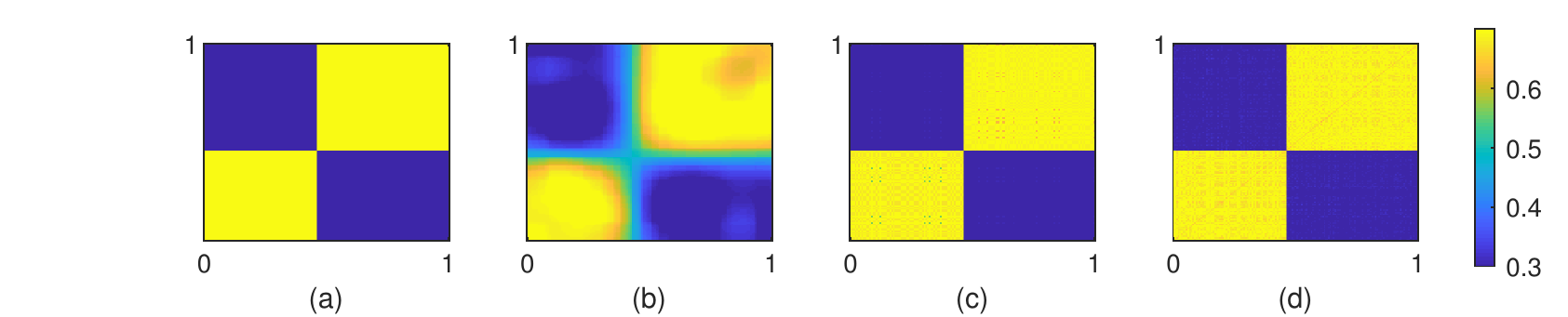} 
\caption{A comparison of estimated graphon matrices for $f_{1},f_{2}, f_{3}$ with $\beta=0$ (replicated networks), in rows $1,2,3$, respectively, where (a) true graphon $f$, (b) proposed methodology, (c) SBA of~\cite{Airoldi13} and (d) NBS of~\cite{Zhang15nbd}. Here $n=150$ and $m=100$.}\label{fig:grnc_es}
\end{figure}

\subsection{Heterogeneous networks ($\beta>0$)}\label{sec:heter}
We report simulation results for the general setting of cross-sectional networks observed with network-level covariates  $\check{z}_{1},\ldots,\check{z}_{m}$. 
Two oracle settings are considered: (i) oracle $1^{\prime}$ informing order statistics $(i)$ of the true node-specific positions $i$, i.e., such that $x_{(1)}\leq x_{(2)}\ldots\leq x_{(n)},$ and the true network-specific positions $z_{1},\ldots,z_{m}$, and (ii) oracle $2^{\prime}$ which again informs order statistics of the true node-specific positions exactly as oracle $1^{\prime}$, however, gives no information on the network specific positions. 
Under both oracles $\hat{x}^{*}_{(i)}=i/(n+1), \forall i \in [n]$ provide oracle estimates of nodal positions, and our algorithm for multi-graphon estimation reduces to nonparametric regression using $\hat{x}^{*}_{(1)},\ldots,\hat{x}^{*}_{(n)}$, and with the exact network positions $z_{1},\ldots,z_{m}$ under oracle $1^{\prime}$, whereas with network-level covariate measurements $\check{z}_{1},\ldots,\check{z}_{m}$ under oracle $2^{\prime}$. Thus, oracle $1^{\prime}$ indicates the best case performance which could be achieved for finite samples if the true set of neighboring nodes were observed, however with imperfect nodal locations $\hat{x}^{*}_{(i)}$.
Oracle $2^{\prime}$ indicates the increase in error (over oracle $1^{\prime}$) resulting from the use of network-level covariates $\check{z}_{l}$ instead of the true network positions ${z}_{l}$.

A comparison of our multi-graphon estimates with existing methods using MSE averaged over $50$ replications is displayed in~\cref{tab:simwc}; visual comparisons of estimates from the proposed method, SBA and NBS are displayed in~\cref{fig:grwc_es}.
Unlike $f_{3},$ for $f_{1}$ and $f_{2}$, the estimated nodal positions implied the same structure as of the true $[f(x_{(i)},x_{(j)};z_{l})]_{i,j,l}$ suggesting that the purified flattened graphon $\bar{f}^{\ast}$ in these cases is identical to the actual flattened graphon $\bar{f}$. Intuitively, this is what we expect given the smooth structure of $f_{1}$, $f_{2}$, and the mixed structure of $f_{3}$. To allow comparisons for $f_{3}$, we plot our proposed estimate of $f_{3}$ with rows and columns permuted to match the true node ordering, i.e. $[f_{3}(x_{(i)},x_{(j)};z_{l})]_{i,j,l}$.

\Cref{tab:simwc} reports MSEs of the multi-graphon array averaged separately for networks generated with weak and strong network-specific effects, precisely, $\{l\in[m]:z_{l}<0.8\}$ and $\{l\in[m]:z_{l}\geq0.8\}$, respectively. From this table, we note the relatively higher MSEs of estimates under oracle $2^{\prime}$ in comparison to oracle $1^{\prime}$. This increase in MSE results from the use of covariates employed as noisy measurements for unobserved network positions, as expected. 
Further, comparing MSEs of oracle $2^{\prime}$ estimates with actual non-oracle estimates across $f_{1}, f_{2}$ and $f_{3}$, it is apparent that $f_{1}$ suffers the most due to relatively poor estimation of latent nodal positions. 
As discussed earlier, this is due to it's extremely smooth structure.
Further, we see that 
our method leads to notably lower MSE for $f_{1}$ in all cases except when $n=50$. Due the smooth structure of $f_1$, $n=50$ nodes prove insufficient for nodal position estimation resulting in a higher MSE.
 For $f_{2}$, our method consistently leads to the smallest MSE, with SBA and NBS leading to the second best performance. 
For $f_{3}$, our proposed estimator is comparable to the best performing approaches of USVT, NBS and SBA for networks with stronger network-specific effects (higher values of $z$) but has a relatively higher MSE otherwise. This is due to the fact that $f_{3}$ is simply a discrete block model for smaller values of $z$ and thus estimation with splines as basis functions even with the true nodal locations does not lead to improved estimation. This is apparent from the MSEs corresponding to the oracle settings of the proposed method which also have higher MSEs for smaller values of $z$. Noting the good performance of NBS for $f_{3}$, we recommend using our approach with kernel regression (e.g. with a uniform kernel) rather than splines, for multi-graphon estimation of networks with discrete structure.

\begin{figure}
\centering
\includegraphics[width=0.57\textwidth]{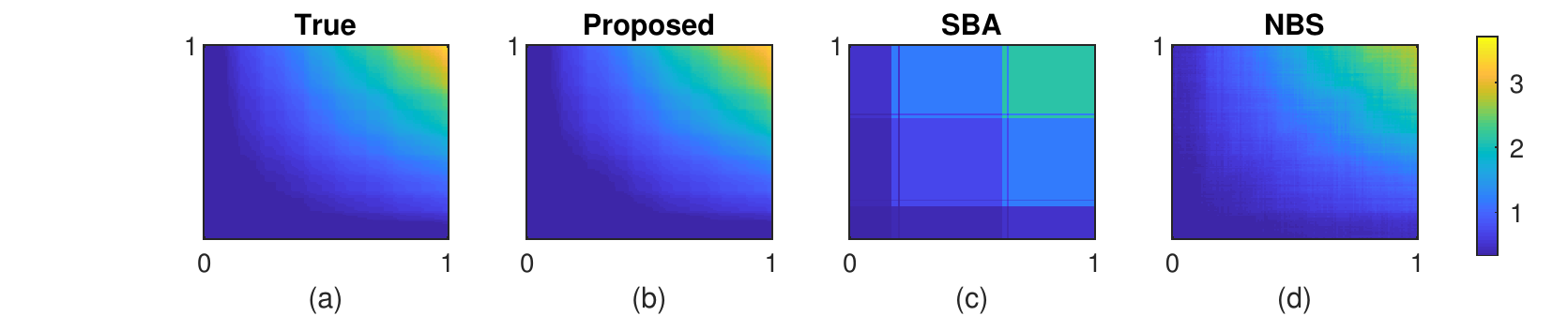}
\bigbreak
\includegraphics[width=.57\textwidth]{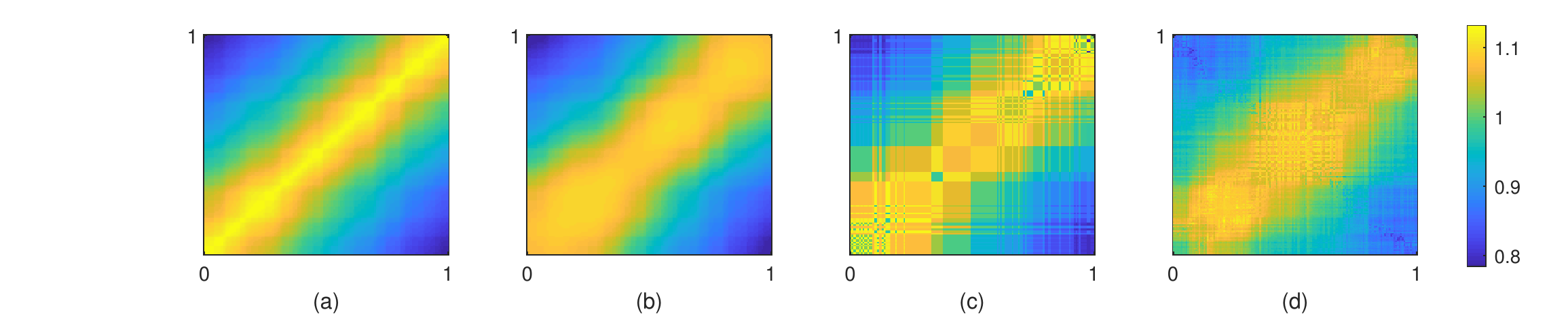}
\bigbreak
\includegraphics[width=.57\textwidth]{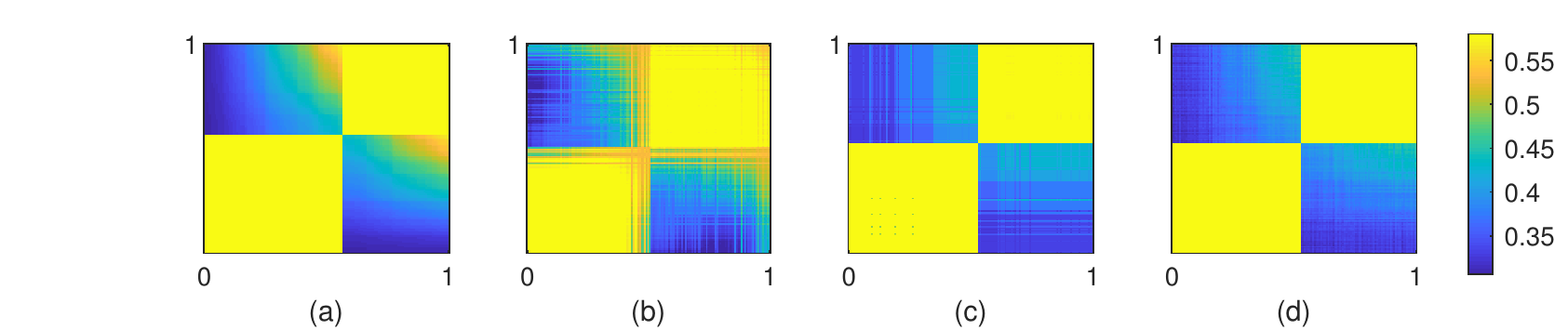}
\caption{Estimated multi-graphon matrices for $f_{1}$ (row 1), $f_{2}$ (row 2), and $f_{3}$ (row 3) with $\beta>0$ (heterogeneous networks) at a fixed network position $z=c$, where (a) true multi-graphon $f(.,.;z=c)$, (b) proposed methodology, (c) SBA of~\cite{Airoldi13} and (d) NBS of~\cite{Zhang15nbd}. Here $n=150$ and $m=100$.}\label{fig:grwc_es}
\end{figure}

\begin{table}
\caption{Mean squared error ($\pm$ std. dev.) comparisons of (multi-)graphon estimates of $f_{1},f_{2},f_{3}$ with $\beta>0$, all multiplied by $10^3$, averaged over $50$ replications. Proposed$^{*}$ (proposed under oracles $1^{\prime}$ and $2^{\prime}$ ), SBA of~\cite{Airoldi13}, SAS of~\cite{Chan14}, USVT of~\cite{Chatterjee15} and NBS of~\cite{Zhang15nbd}. MSE for multi-graphon estimates from the proposed method are averaged for $\{f_{..l}, z_{l}<0.8\}$ and $\{f_{..l}, z_{l}\geq 0.8\}$.
}
 \label{tab:simwc}
\begin{center}
\resizebox {\textwidth }{!}{
 \begin{tabular}{l l l| l l l l l l l l l l }
& & & \multicolumn{2}{c}{Proposed$^{*}1^{\prime}$} & \multicolumn{2}{c}{Proposed$^{*}2^{\prime}$} & \multicolumn{2}{c}{Proposed} & SBA & SAS & USVT & NBS \\
 $f$  & $n$ & $m$ & $z<0.8$ & $z\geq0.8$    &  $z<0.8$ & $z\geq0.8$    & $z<0.8$ & $z\geq0.8$ &  &  &  &  \\
 \hline\\
 \multirow{ 6}{*}{$f_{1}$} & 50 & 150 & $10.00 (3.50)$ & $7.50 (1.60)$ & $12.10 (4.40)$ & $9.30 (2.60)$ & $160.40 (64.10)$ & $83.80 (23.40)$ & $248.50 (178.80)$ & $75.80 (45.40)$ & $50.10 (23.40)$ & $55.10 (28.80)$ \\
    & 100 & 150 & $6.30 (2.40)$ & $3.90 (0.63)$ & $7.90 (3.10)$ & $5.90 (1.70)$ & $14.50 (5.30)$ & $10.20 (2.00)$ & $159.80 (164.10)$ & $ 58.40 (38.80)$ & $37.90 (30.0)$ & $37.40 (27.20)$  \\
     &150 & 150 &$4.40 (1.70)$ & $2.60 (0.38)$ & $5.80 (2.70)$ & $3.80 (1.70)$ & $12.80 (5.80)$ & $7.30 (2.10)$ & $157.30 (151.10)$ & $38.90 (22.70)$ & $32.40 (23.90)$ & $30.80(20.10)$\\
   &150 & 50 & $3.80 (1.80)$ & $2.40 (0.48)$ & $6.50 (4.10)$ & $5.40 (1.60)$ & $17.30 (7.20)$ & $11.50 (2.20)$ & $91.70 (88.60)$ & $40.40 (22.40)$ & $34.50 (26.00)$ & $33.40 (21.30)$ \\
   & 150 & 100 & $4.20 (1.70)$ & $2.40 (0.27)$ & $6.00 (3.30)$ & $4.00 (1.40)$ & $15.00 (6.30)$ & $9.10 (2.10)$ & $130 (132.70)$ & $40.30 (22.70)$ & $34.00 (25.20)$ & $32.10 (21.10)$ \\
   &150 & 150 & $4.40 (1.70)$ & $2.60 (0.38)$ & $5.80 (2.70)$ & $3.80 (1.70)$ & $12.80 (5.80)$ & $7.30 (2.10)$ & $157.30 (151.10)$ & $38.90 (22.70)$ & $32.40 (23.90)$ & $30.80(20.10)$\\
	\hline\\	
\multirow{ 6}{*}{$f_{2}$} & 
 50 & 150 & $1.70 (0.15)$ & $1.60 (0.10)$ & $1.90 (0.31)$ & $1.80 (0.20)$ & $3.60 (0.60)$ & $3.00 (0.33)$ & $4.80 (1.70)$ & $6.70 (1.70)$ & $5.80 (1.40)$ &$4.90 (1.40)$  \\
    &100 & 150 & $0.36 (0.08)$ & $0.30 (0.02)$ & $0.41 (0.15)$ & $0.35 (0.04)$ & $0.46 (0.16)$ & $0.38 (0.04)$ & $2.90 (1.80)$ & $6.30 (1.70)$ & $5.90 (1.50) $ & $2.50 (1.00)$ \\
    &150 & 150 & $0.35 (0.08)$ & $0.28 (0.02)$ & $0.40 (0.14)$ & $0.35(0.04)$ & $0.41 (0.14)$ & $0.36 (0.04)$ & $3.00 (1.90)$ & $5.60 (1.50)$ & $5.50 (1.40)$ & $1.50 (0.74)$ \\
   &150 & 50 & $0.36(0.08)$ & $0.30 (0.03)$ & $0.40 (0.14)$ & $0.31(0.03)$ & $0.42 (0.14)$ & $0.33 (0.04)$ & $2.40 (1.30)$ & $5.60 (1.50)$ & $5.50 (1.30)$ & $1.60 (0.78)$ \\
   &150 & 100 & $0.36 (0.08)$ & $0.30 (0.02)$ & $0.41 (0.15)$ & $0.33 (0.03) $ & $0.46 (0.16)$ & $0.38 (0.04)$ & $1.80 (1.30)$ & $5.60 (1.50)$ & $5.50 (1.30)$ & $1.70 (0.85)$  \\
    &150 & 150 & $0.35 (0.08)$ & $0.28 (0.02)$ & $0.40 (0.14)$ & $0.35(0.04)$ & $0.41 (0.14)$ & $0.36 (0.04)$ & $3.00 (1.90)$ & $5.60 (1.50)$ & $5.50 (1.40)$ & $1.50 (0.74)$ 
 \\
	\hline\\
\multirow{ 6}{*}{$f_{3}$} & 50 & 150 & $6.40 (3.30)$ & $2.70 (0.55)$ & $6.70 (3.40)$ & $2.90 (0.13)$ & $7.80 (4.00)$ & $3.30 (0.70)$ & $2.90 (2.00)$ & $13.40 (9.50)$ & $2.40 (4.40)$ & $2.20 (1.80)$ \\
    & 100 & 150 & $4.70 (2.60)$ & $2.10 (0.55)$ & $4.90 (2.70)$ & $2.30 (0.18)$ & $5.80 (3.40)$ & $2.50 (0.63)$ & $51.30 (6.50)$ & $53.60 (5.90)$ & $50.70 (6.50)$ & $51.00 (6.50)$  \\
    &150 & 150 & $4.00 (2.10)$ & $1.80 (0.44)$ & $4.10 (2.40)$ & $1.90 (0.46)$ & $4.90 (2.90)$ & $2.10 (0.47)$ & $2.40 (1.70)$ & $13.90 (9.00)$ & $2.20 (1.60)$ & $2.20 (1.59)$  \\
   &150 & 50 & $4.20(2.20)$ & $1.90 (0.45)$ & $4.40 (2.30)$ & $2.10 (0.27)$ & $6.20 (3.20)$ & $3.00 (0.64)$ & $2.40 (1.70)$ & $15.80 (9.90)$ & $2.10 (1.70)$ & $2.20 (1.70)$ \\
    & 150 & 100 & $4.00 (2.10)$ & $1.80 (0.39)$ & $4.20 (2.30)$ & $2.00 (0.41)$ & $5.40 (3.00)$ & $2.50 (0.45)$ & $2.30 (1.90)$ & $13.80 (8.70)$ & $1.90 (1.60)$ & $2.00 (1.60)$  \\
    &150 & 150 & $4.00 (2.10)$ & $1.80 (0.44)$ & $4.10 (2.40)$ & $1.90 (0.46)$ & $4.90 (2.90)$ & $2.10 (0.47)$ & $2.40 (1.70)$ & $13.90 (9.00)$ & $2.20 (1.60)$ & $2.20 (1.59)$  \\
\end{tabular}}
\end{center}
\end{table}

\section{Two data examples}\label{sec:data}
We illustrate the performance of the proposed multi-graphon estimator using two publicly available data sets: 
 (i) a dynamic contact network of ants~\cite{Mersch13}, and (ii) a human connectome dataset named Templeton-114~\cite{Roncal2013, Kiar2016ndmg}.
\subsection{Dynamic contact network of ants}
With a view to understand division of labor among ant workers, movements in six colonies of the ant \textit{Camponotus fellah} were tracked over a period of $41$ days with network interactions between any two ant workers (nodes) determined by their physical proximity (see SI~\cite{Mersch13} for more details).
We illustrate our methodology using data from colony $3$ which has the maximum number of overlapping ant workers (precisely $n=96$) over the duration of $m=41$ days. This leads to $41$ adjacency matrices, each of size $96 \times 96$, i.e., $\{A_{ijt}, i,j\in[n]\times [n], t\in [m]\}$ with $A_{ijt}$ denoting the count of interactions between ants $i$ and $j$ on day $t$. Using behavioral signatures of ant workers such as visits to the brood, foraging trips and visits to the rubbish pile, each ant worker is also recorded to be 
a nurse (N), or a forager (F) or a cleaner (C), respectively, across four consecutive time periods, each of approximately $10$ days~\cite{Mersch13}. 

\begin{figure}[t]
\centering
\includegraphics[width=0.6\textwidth]{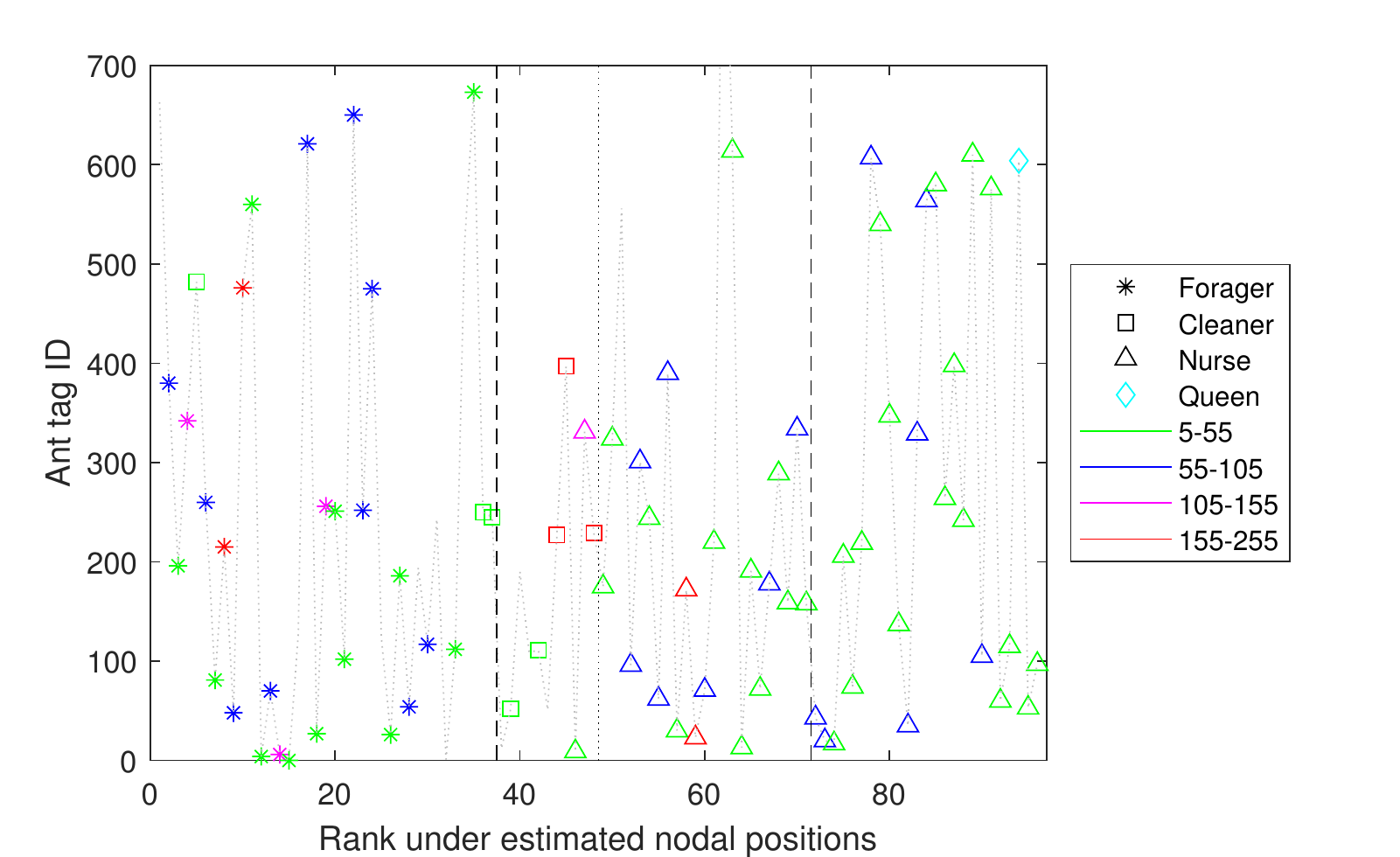}
\caption{Ant tag ids (y-axis) against their ranks (x-axis) based on the estimated nodal positions.
The dotted line in the middle is plotted for reference; dashed lines are used to interpret results (details in text). See legend for occupation and age group of each ant.}
\label{fig:estdpermAnts3}
\end{figure}

\cref{fig:estdpermAnts3} displays a rearrangement of the ant workers (nodes) sorted by increasing nodal positions (x-axis) estimated following the proposed~\cref{alg:meta}, against their original ant indices as recorded in the data set. 
Ant worker attributes such as their majority occupation (over the four time periods) and age group are also displayed (see legend in \cref{fig:estdpermAnts3}). The dotted line in the middle, plotted for reference, divides the set of nodes into two equal groups on each side. We note a clear spatial segregation with the 
forager ants (`$\ast$') always positioned to the left of the
 dotted line and a majority of the nurse ants (`$\triangle$') positioned to the right of the dotted line;
cleaner ants (`$\square$') are clearly positioned in between these two larger occupational groups.
This suggests that ant workers with the same occupation were estimated to be closer to each other than ants with different occupations via the distance estimation approach. Such a spatial segregation is clearly not implied by the age attribute, as ants from the same age group are not always positioned closer to each other. 
Further, we note that the queen ant (`$\textcolor{blue}{\diamond}$') is estimated to be spatially closer to the group of nurses (`$\triangle$') and is positioned far from cleaner and forager groups. This is in agreement with the well-known behavior of queen ants who are solely responsible for reproduction. 

\begin{figure}
\centering
\includegraphics[width=\textwidth]{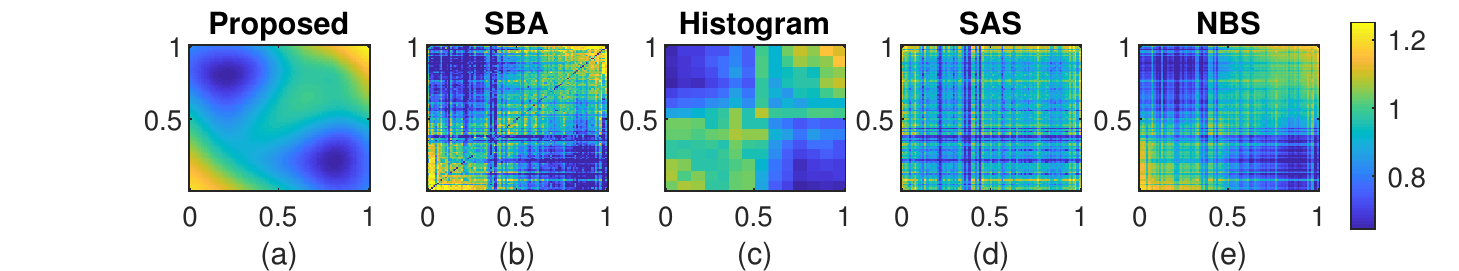}
\caption{Graphon estimates $\hat{f}^{1/4}$ for contact network of ants, assuming i.i.d networks over time, where: (a) the proposed methodology, (b) SBA of~\cite{Airoldi13}, (c) network histogram of~\cite{olhede14}, (d) SAS of~\cite{Chan14}, and (e) NBS of~\cite{Zhang15nbd}. For comparison, estimates from SBA, SAS, and NBS were re-arranged to correspond to increasing nodal position estimates from our algorithm.
The power root stabilizes the variance of the intensity displayed using the color spectrum and is solely for ease of visualization.}
\label{fig:kernelestAnts3}
\end{figure}

We first study comparisons for graphon estimates obtained under the assumption of an i.i.d (or replicated) collection of networks over time.
 The result from our approach and comparisons with existing techniques applied to the aggregated adjacency $\bar{A}$ 
are displayed in~\cref{fig:kernelestAnts3}, where for convenience of  
comparisons, rearranged matrix estimates of SBA, SAS and NBS, with nodes sorted by increasing nodal positions estimated from our approach, are shown.
We see a good agreement between the proposed estimator and all other methods except SAS, 
with high intensity regions at the edges of the main diagonal (corresponding to subgroups of forager and nurse nodes); and relatively low intensity of connection along the off-diagonal.
\begin{figure}[hbt!]
\centering
\includegraphics[width=0.50\textwidth]{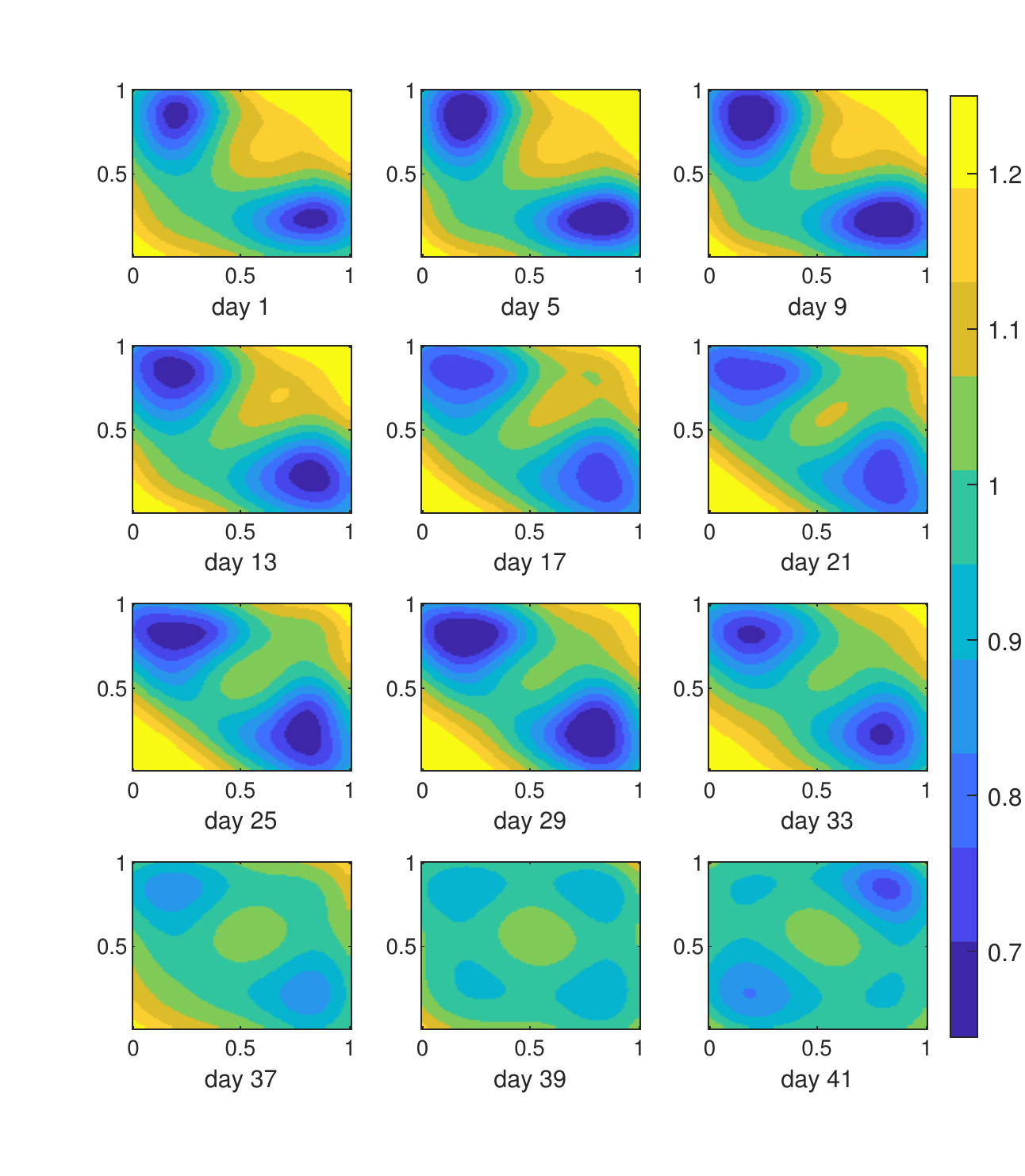}
\caption{Multi-graphon matrix estimates $\hat{f}^{1/4}(.,.,t_{l}), t_{l}=l/41, l \in [41]$ using the proposed method for day $l$ shown along the x-axis. The power root stabilizes the variance of the intensity displayed using the color spectrum and is solely for ease of visualization.}
\label{fig:Gr_est_Ant3}
\end{figure}

Existing studies on organizational behavior of ants such as~\cite{Miele17,Mersch13} and references therein, suggest that the assumption of identically distributed networks over time is unrealistic for such ant interaction data. Our multi-graphon estimates displayed in~\cref{fig:Gr_est_Ant3} indicate that this is indeed the case as newer structural features become apparent when estimation is performed without assuming identically distributed networks over time. \Cref{fig:Gr_est_Ant3} shows how the network structure changes over the duration of $41$ days with a significant decrease in intensity of interactions towards the end of the period, particularly, beyond day $37$.
More precisely, notably high intensity of interactions are observed until day $33$ of the experiment for a small proportion of nurse ant workers (top right corner of multi-graphon estimates), and forager ant workers (bottom left corner of multi-graphon estimates), beyond which intensity of interactions in these regions begins to decrease. 
In fact, the highest intensity of interaction by the end of the experiment is between cleaners and a small subset of forager and nurse ants: precisely the set of ant workers positioned 
within the dashed lines displayed in~\cref{fig:estdpermAnts3}. 

\begin{figure}[hbt!]
\centering
\includegraphics[width=0.55\textwidth]{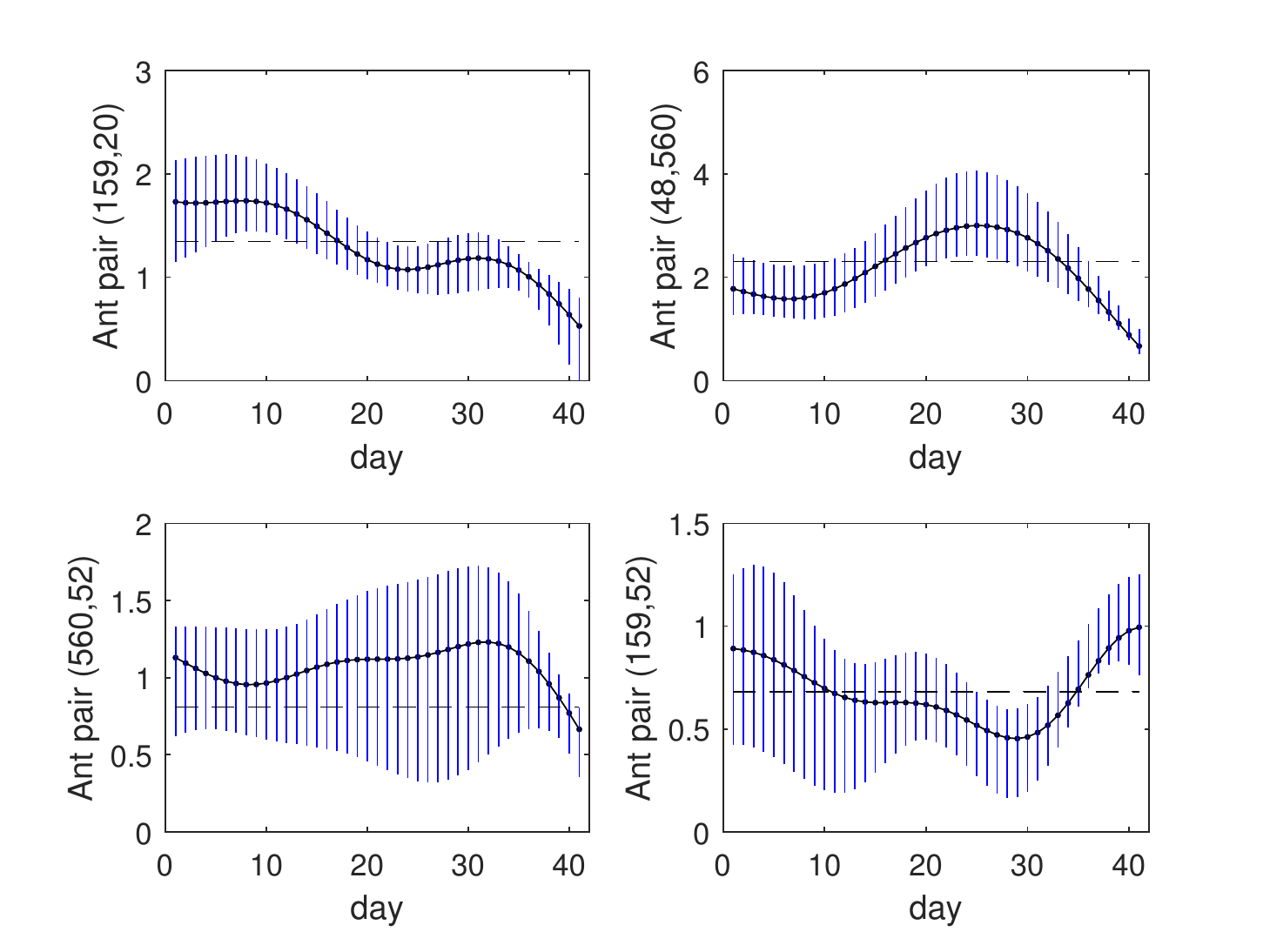}
\caption{$95\%$ confidence intervals for estimated intensity of pairwise interactions between ant pairs over time with original ant tag ids and occupations given by: first row, left: $(159,20)$, N-N; first row, right: $(48,560)$, F-F; second row, left: $(560,52)$, F-C; second row, right $(159,52)$, N-C.}
\label{fig:Ant3_ci}
\end{figure}

Estimates of pairwise intensity of interactions for four pairs of ants and the corresponding $95\%$ confidence intervals over time, obtained via subsampling bootstrap are displayed in~\cref{fig:Ant3_ci}. Significant differences in interaction behavior over time periods are easily identified for all pairs of ants except the forager-cleaner (F-C) ant pair $(560,52)$ (or $(15,44)$ in the estimated ordering), where confidence bars overlap across all days. For example, for the forager-forager (F-F) ant pair $(48,560)$ displayed in the second subplot, the estimated intensity of interaction over the first $10$ days is significantly lower in comparison with intensity over days $20-30$, decreasing again beyond day $36$. 
It is interesting to note that the intensity of interaction between the nurse-cleaner (N-C) pair $(159,52)$ over days $36-41$ is significantly higher than the intensity over days $25-31$, suggesting a change in behavior somewhere between these two time periods.
Noting the occupation of the nurse ant worker $159$, we find that it is recorded to be a nurse in the first three periods of data collection (precisely, days $1-31$) and a cleaner for the last period spanning days $32-41$. This could be a possible explanation for the significant increase in intensity between the N-C pair $(159,52)$ with days $32-35$ corresponding to a transition period for a change in occupation from a nurse to a cleaner.

\subsection{Human connectome data}
This data set comprises of structural brain networks on $n=116$ brain regions, known as regions of interest (ROIs), observed for $m=256$ subjects. For each subject $l\in [m]$, the existence of an edge between brain regions $i\in [n]$ and $j \in [n]$ is determined from multimodal magnetic resonance imaging data~\cite{Kiar2016ndmg}, and corresponds to the presence of atleast one white matter fiber connecting the two regions, (see~\cite{gray2012magnetic,Roncal2013} for details).
The brain regions considered in this data set are given by the Automated Anatomical Labeling (AAL 116) cortical atlas~\cite{tzourio2002}. 
This data set also includes a creativity score for each subject, measured via the composite creativity index (CCI) of~\cite{Jung2010}. 
The CCI scores are informed by 
ranks assigned to the creative products of each subject by three independent judges.

\begin{figure}[hbt!]
\centering
\includegraphics[width=0.56\textwidth]{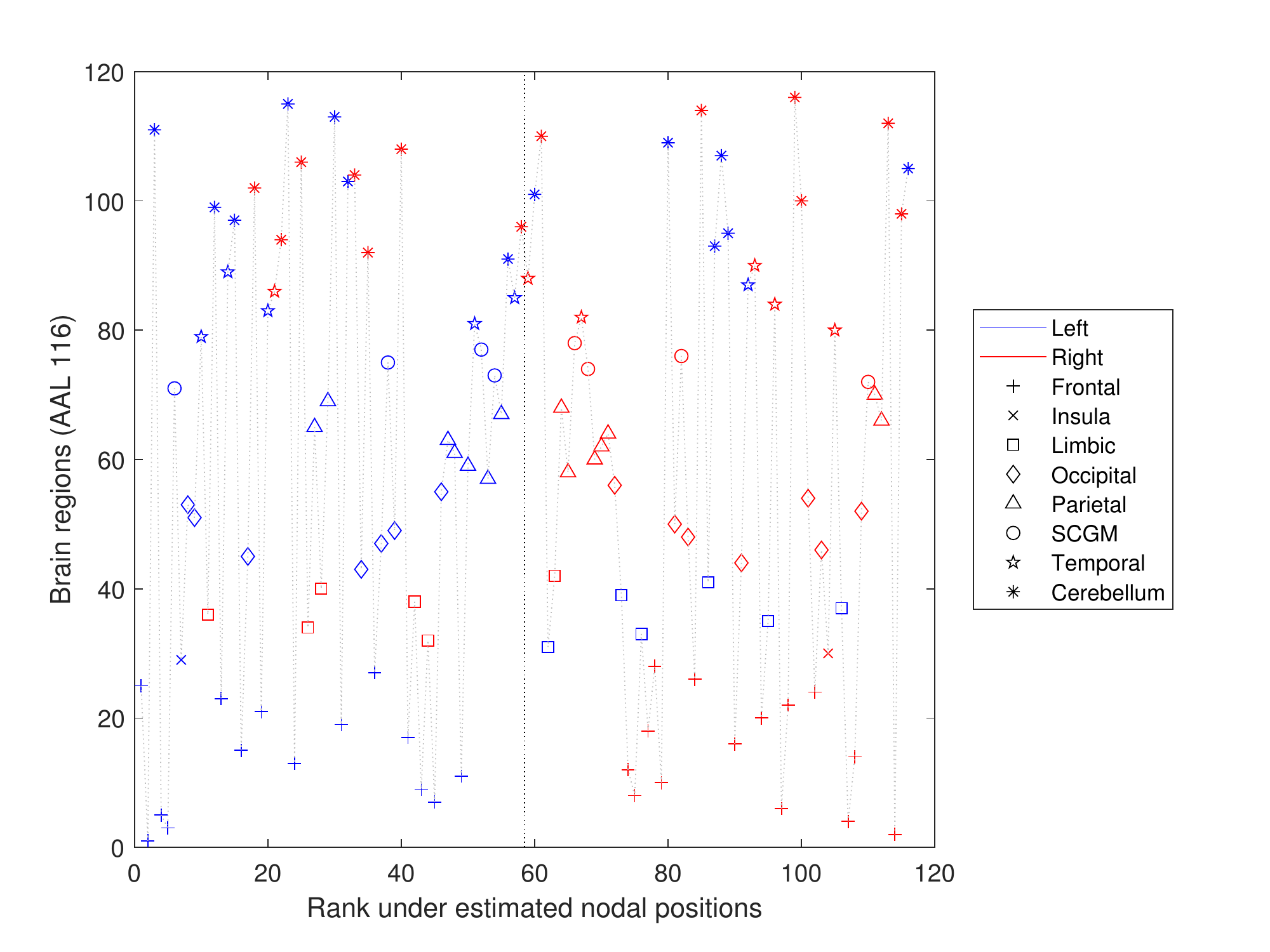}
\caption{Brain regions (indices from AAL116 atlas, y-axis) 
against their ranks (x-axis) based on the estimated nodal positions.
The dotted line in the middle is plotted for reference (see text).}
\label{fig:permxihatBrain}
\end{figure}

\Cref{fig:permxihatBrain} displays brain regions sorted by increasing nodal embedding position estimates (x-axis) from our proposed algorithm against their actual AAL116 indices (y-axis). The membership of each brain region in one of the two hemispheres--left or right, and one of the eight lobes-- Frontal, Insular, Limbic, Occipital, Parietal, SCGM, Temporal, Cerebellum, is also displayed (see legend). The dotted line in the middle, plotted for reference, divides the set of nodes into two equal groups on each side. 
Noting the hemisphere (and lobe) membership of nodes on the left and right side of the dotted line in~\cref{fig:permxihatBrain}, we observe that ROIs belonging to the left and right hemispheres, lie to the left and right of the dotted line respectively, for members of all lobes except Limbic ($\Box$) and Cerebellum ($\ast$). Since ROIs belonging to the left and right hemispheres, lie to the left of the origin (negative x-axis) and right of the origin (positive x-axis) in the standard MNI space, respectively, it suggests that nodal positions estimated via our algorithm, for a majority of brain regions are coarsely aligned with their actual spatial coordinates along the first dimension (or x-coordinates). Further we see that nodes from the Limbic lobe are embedded such that its members from the left (right) hemisphere are positioned to the right (left) of the dotted line (centre), whereas for nodes from the Cerebellum lobe, left and right hemisphere members are mixed on either side of the dotted line. 

A comparison of our graphon estimate under the replicated network assumption with estimates from SBA of~\cite{Airoldi13}, network histogram of~\cite{olhede14}, SAS of~\cite{Chan14}, USVT of~\cite{Chatterjee15}, and NBS of~\cite{Zhang15nbd}, 
 is displayed in~\cref{fig:est_nc}. Clearly, network structure is only apparent from the proposed estimate and network histogram of~\cite{olhede14}, a graphon function estimator.
The lack of structural visibility in estimates from all other methods is due to the absence of a meaningful ordering on the set of nodes, typically achieved using node-specific covariates which are not observed in this dataset.
To allow comparison, re-arranged matrix estimates of SBA, SAS, USVT, and NBS with nodes sorted by increasing nodal position estimates from our algorithm, are displayed in~\cref{fig:est_nc1}. Overall, at a coarse level we see a good agreement between estimates from all methods except SAS.
Our estimator clearly indicates assortative community-like behavior for nodes positioned at the two extremes, precisely, nodes with estimated indices $1-35$ and $81-116$ (x-axis of ~\cref{fig:permxihatBrain}).
We see a very high intensity of interaction for nodes within these two groups and very low intensity of interaction across the two groups, and clearly, a relatively weaker community structure for nodes positioned in the middle (node indices $48-74$). 

Assuming the collection of networks from subjects to be non-identically distributed~\cref{fig:Gr_est_CCI} displays multi-graphon estimates obtained with network-level covariates $\check{z}_{l}$ as the normalized (max norm) CCI score of subject $l\in [m]$. 
From these plots it is evident that network structure changes as we go from subjects expressing low creativity to high creativity,
e.g., with significantly different intensities of interactions along the main diagonal with increasing CCI. 
For a closer inspection~\cref{fig:CI} displays $95\%$ confidence intervals and estimates of pairwise intensity of interaction for four different ROI pairs, as a function of
CCI scores. 
An immediate observation is that we may not always observe (a significant) increase in intensity of interaction with increase in creativity levels measured via CCI. This is visible from~\cref{fig:CI} where red bars indicate similar intensities with increase in CCI in subplots (a), (b), (d) and a significant decrease in intensity with increase in CCI in subplot (c).
Secondly, these plots suggest that the CCI score threshold for partitioning network samples into `low' and `high' creativity groups (e.g.~\cite{DD18}) may vary depending on the ROI pairs of interest.
Based on these findings, in practice, we recommend fixing the set of ROIs of interest to the practitioner, to infer a meaningful grouping of network samples
 before performing tasks such as 
identifying a subset of edges which provide evidence of change across low and high creativity groups or classification into categories constructed artificially from continuous-valued information, for example as considered in~\cite{DD18}.
This is crucial as otherwise aggregated behavior of each partition may not be representative of the actual behavior due to significant differences within the chosen subset of network samples, resulting in misleading conclusions.

\begin{figure}[hbt!]
\centering
\includegraphics[width=0.95\textwidth]{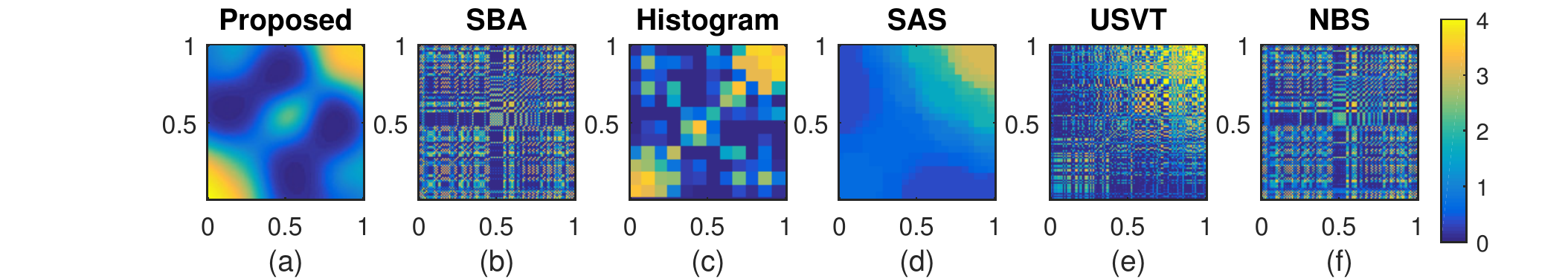}
\caption{Graphon matrix estimates $\hat{f}^{1/2}$ for connectome data, assuming i.i.d networks over subjects, where: (a) the proposed methodology, (b) SBA of~\cite{Airoldi13}, (c) network histogram of~\cite{olhede14}, (d) SAS of~\cite{Chan14}, (e) USVT of~\cite{Chatterjee15} and (f) NBS of~\cite{Zhang15nbd}.}
\label{fig:est_nc}
\bigbreak
\includegraphics[width=.62\textwidth]{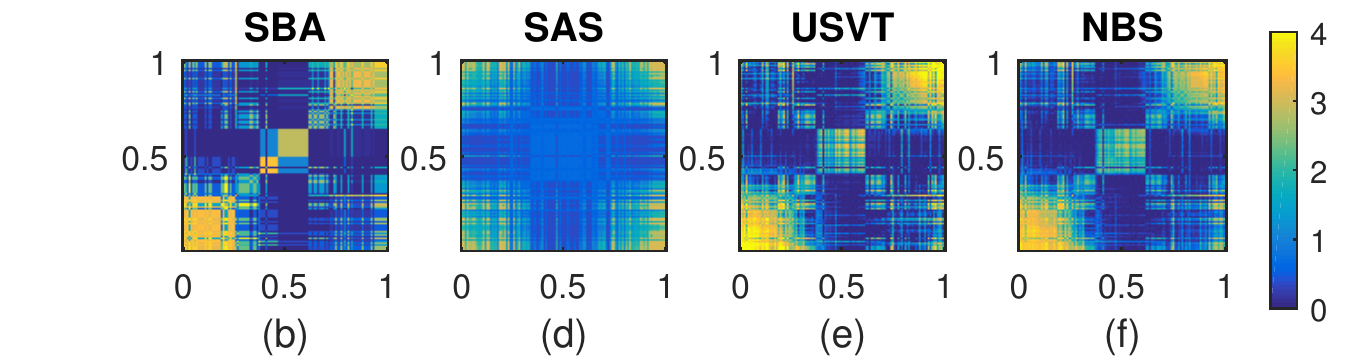}
\vspace{-4mm}
\caption{Re-arranged graphon matrix estimates (from above)
 (b) SBA (d) SAS (e) USVT and (f) NBS, with nodes sorted by increasing nodal positions estimated from our algorithm.}
\label{fig:est_nc1}
\end{figure}

\subsubsection{Application to resampling networks}
Network summary statistics such as 
triangle frequency, average
path length, transitivity, network edge density are of great practical interest and have been studied in the context of brain network organisation and creativity, for example as in~\cite{DD18, Maugis17, bullmore12}.
According to~\cite{bullmore12}, structural brain networks of highly creative individuals are found to exhibit small-world phenomenon with high triangle frequency, low average path length, high edge density, and high transitivity. 
To check if the small-world behavior for high creativity individuals suggested by previous studies, is a feature implied by our multi-graphon estimate, we study network summaries for samples ${A}$ generated using the estimated multi-graphon $\hat{f}$. 
For a given normalized creativity score $\check{z}\in (0,1),$ we generated $B$ networks $A_{1}(\check{z}),\ldots,A_{B}(\check{z})$, each of size $n \times n$ ($n=116$), as independent Bernoulli trials where, $A_{ij}(\check{z})\sim \text{Bernoulli}(\hat{\rho}_n\hat{f}(\hat{x}_{i},\hat{x}_{j};\check{z}))$, for $(i,j) \in [n]\times [n]$.
Subsequently, the four network statistics --triangle frequency, average
path length, transitivity, and network edge density were computed for each $A_{1}(\check{z}),\ldots,A_{B}(\check{z})$. 
\Cref{fig:smallworld_CCI} displays the corresponding $95\%$ confidence intervals for these four network statistics with increasing creativity scores, obtained using $B=10000$ networks. 
From these plots, differences in network statistics across creativity levels are apparent.
Further, we see that triangle frequency, edge density, and transitivity are significantly higher, whereas average path length, is significantly lower for subjects with high creativity in comparison to those with low creativity, confirming the small-world phenomenon for high creativity brains~\cite{bullmore12}.

\begin{figure}[hbt!]
\centering
\includegraphics[width=0.5\textwidth]{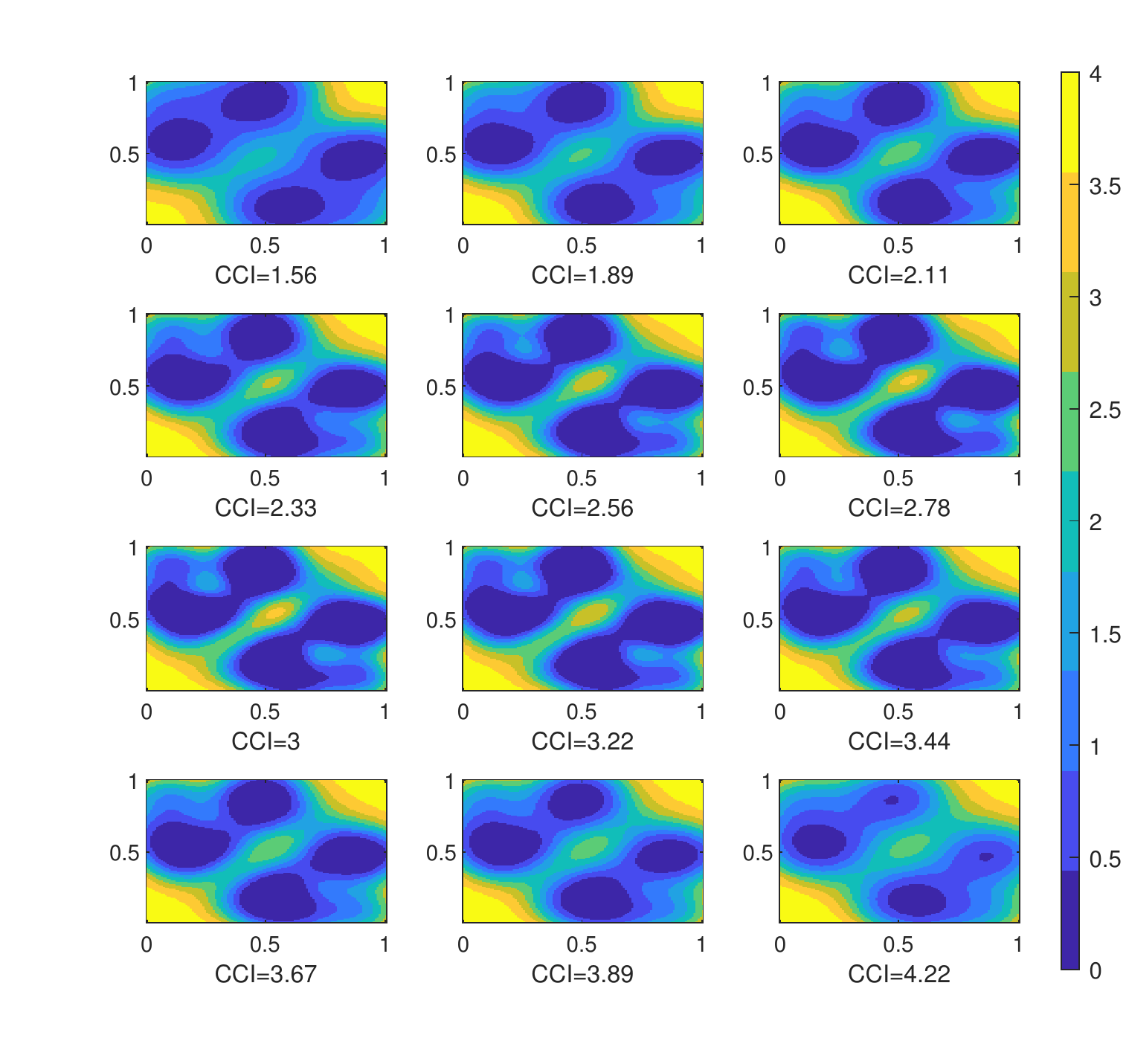}
\caption{Multi-graphon matrix estimates $\hat{f}^{1/2}(:,:,{z}_{l})$ using the proposed method for increasing CCI scores shown along the x-axis. The power root stabilizes the variance of the intensity displayed using the color spectrum and is solely for ease of visualization.}
\label{fig:Gr_est_CCI}
\end{figure}

\begin{figure}[hbt!]
\centering
\includegraphics[width=0.52\textwidth]{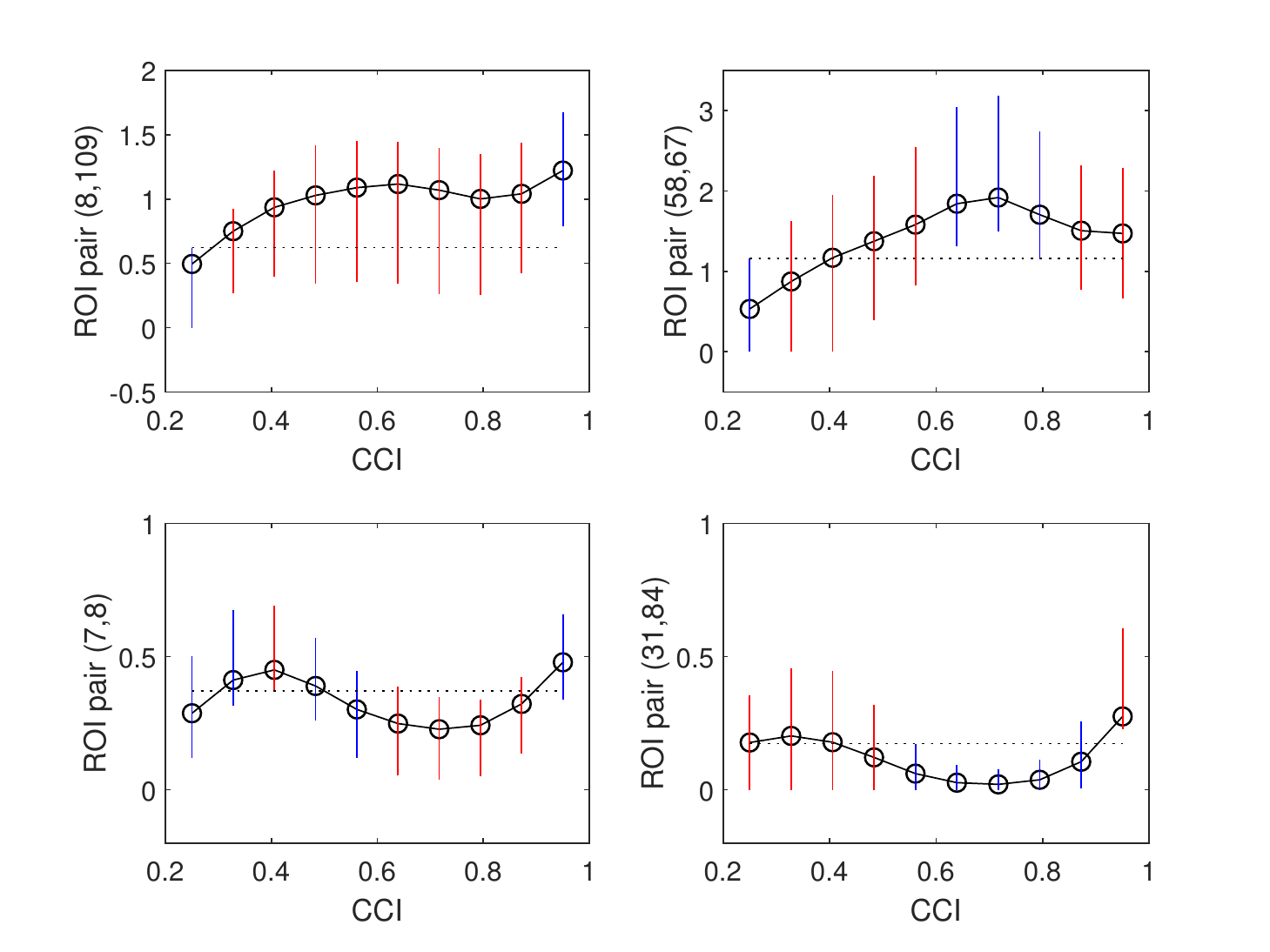} 
\caption{Multi-graphon estimates $\hat{f}^{1/2}$ and $95\%$ confidence intervals for node pairs (indices correspond to AAL116 atlas) with increasing CCI scores. First row, left: $(8,109)\equiv$ ( Frontal Mid.(R), Vermis12), first row, right: $(58,67)\equiv$ (Postcentral(R), Precuneus(L)), second row, left: $(7,8)\equiv$ (Frontal Mid(L) , Frontal Mid(R)) and second row, right: $(31,84)\equiv$ (Cingulum Ant.(L),Temporal Pole Sup(R)). The red bars are used to visualize changes (or no change) in intensity with increasing CCI (see text for details).}
\label{fig:CI}
\end{figure}

\begin{figure}[t]
\centering
\includegraphics[width=0.55\textwidth]{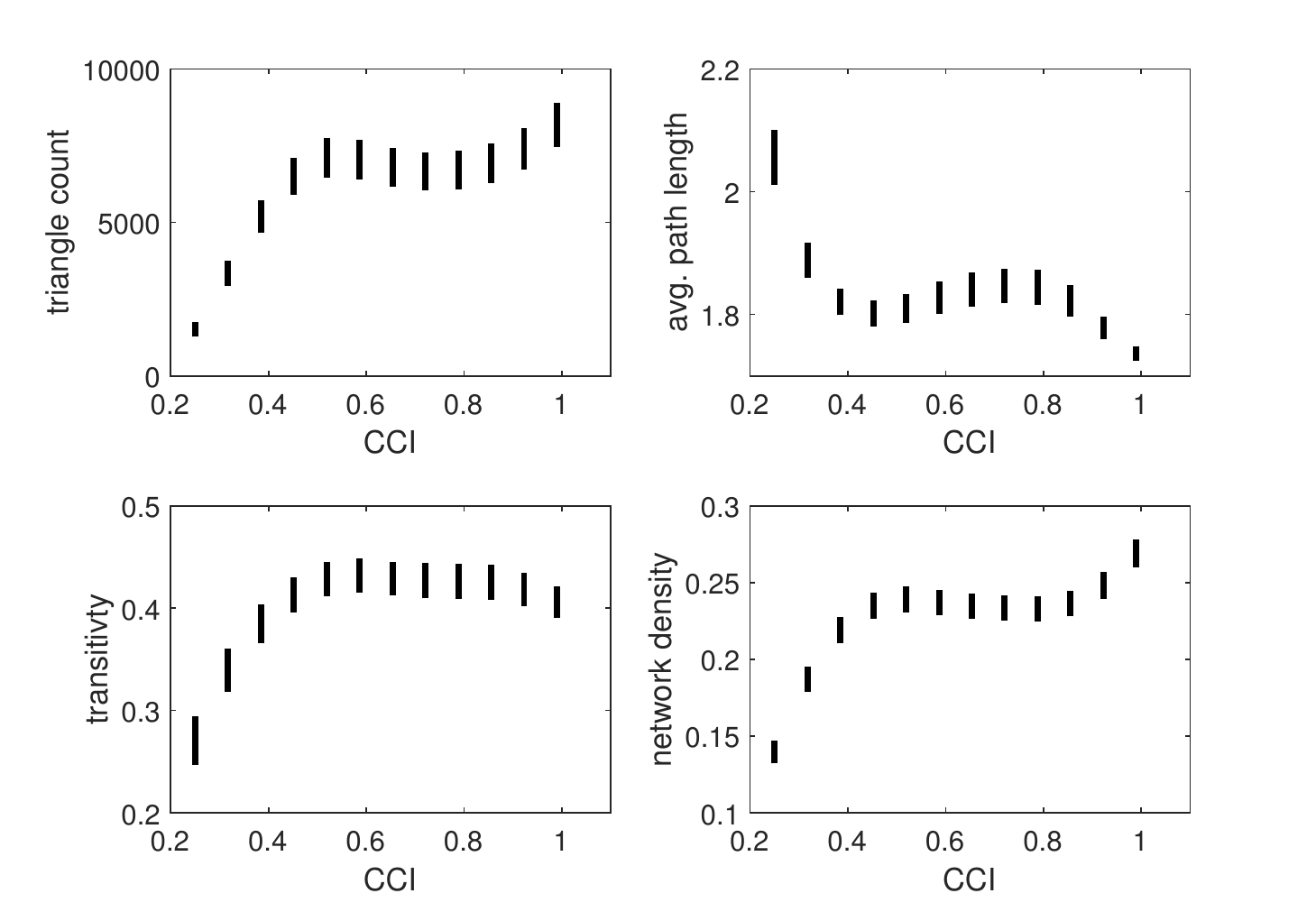}
\caption{$95\%$ bootstrap confidence intervals for different network summary statistics (y-axis) as a function of CCI scores. These were obtained via networks resampled using our estimated multi-graphon with $B=10000$ bootstrap replications for each CCI score.}
\label{fig:smallworld_CCI}
\end{figure}

\section{Conclusion}
By establishing regimes under which ordinal embedding allows consistent estimation of latent nodal positions in the purified graphon space, we have shown how standard smoothing techniques (kernel methods, regression splines and others) can be employed for estimation of the network generating process. We achieved this for a collection of networks on the same set of nodes, which are commonly observed in many applications. 
With these results, estimation of the multi-graphon model from a set of networks observed over time simply reduced to nonparametric regression with estimated nodal positions and equi-spaced time points. For cross-sectional networks, the same was achieved using network-level covariates as measurements for unobserved network-positions. 
In applications where repeated measurements on each of the $m$ networks are available, one may follow the approach outlined in this paper to likewise define pairwise distance between networks to allow estimation of latent network-positions. 

Further, our approach may be used as a building block to study richer models describing network effects through the multi-graphon function. For example, with the multi-graphon function modeled as the sum of a standard two-dimensional graphon function and with either $p$ scalar functions of $p$ covariates as in an additive model~\cite{Hastie} or with a simple linear combination of $p$ covariates implying a partially linear model~\cite{Carroll}. These models shall allow one to integrate more than a single network-level covariate to explain variability across networks without having to deal with the curse of dimensionality via the multi-graphon function.
Modeling and estimation techniques developed in this paper may be extended to longitudinal networks to simultaneously estimate structural variability across both the subject and time axes, as we intend to do in future work.

\section{Acknowledgements}
The authors thank Dr. Joshua T. Vogelstein and Eric Bridgeford at John Hopkins University for sharing the human connectome data. We are also grateful to Professor Carey Priebe for helpful discussions. 

\begin{appendix}
\section{Proofs}\label{app:Proofs}
To prove the main results in~\cref{latent-est-rate,main-clt} we first consider the following result on consistency of pairwise distance estimates. We show that under the null of~\cref{model-def} the estimator produced by~\cref{alg:dist} is a consistent estimator of ${\mathrm{dist}}(\bar{f})$ given by~\cref{def:dist}.

\begin{proof}[Proof of~\cref{prop:consistency}]
Set $U\sim\mathrm{U}[0,1]$. The proof proceeds by computing the variances.
Note here that while we could have proceed like~\cite[Theorem 1.]{Airoldi13} (i.e., via Bernstein's inequality) we found that inefficient when aiming to account for sparsity and varied speed for the growth of $m$ relative to $n$.
To do so, we first consider the $\hat s_{ik} = \sum_{l\in S} A_{ikl}/|S|$, for fixed $i, k$. There, we see that conditionally on $x_i, x_k$, $(A_{ikl})_l$ is i.i.d. $\mathrm{Bernoulli}\big(\rho_n f(x_i,x_k;U)\big)$, so that $\hat s_{ik} \sim  \mathrm{Binomial}(|S|, \rho_n\E_U f(x_i,x_k;U))/|S|$ with $U$ the uniform distribution on $[0,1]$. Therefore, we have that conditionally on $x_i, x_{k'}, x_k$
\begin{align*}
\E\hat s_{ik}
  & = \rho_n\E_U f(x_i,x_k;U) = \rho_n\bar f(x_i, x_k)\\
\var\hat s_{ik}
  & = {\rho_n\bar f(x_i, x_k)(1-\rho_n\bar f(x_i, x_k))}/{|S|} = \Theta\left(\frac{\rho_n}{m}\right)\\
\cov(\hat s_{ik}, \hat s_{ik'})
  & = \E(\hat s_{ik}\hat s_{ik'}) - \E\hat s_{ik} \E\hat s_{ik'}\\
  & = \frac{1}{|S|^2}\sum_{l,l'\in S} \E A_{ikl}A_{ik'l'} - {\rho_n^2\bar f(x_i, x_k)\bar f(x_i, x_k')}\\
  & = \frac{1}{|S|^2}\left(\sum_{l\neq l'\in S} \E A_{ikl}\E A_{ik'l'} + \sum_{l\in S}\E A_{ikl}A_{ik'l}\right) - {\rho_n^2\bar f(x_i, x_k)\bar f(x_i, x_k')}\\
   & = \frac{|S|(|S|-1)\rho_n^2\bar f(x_i, x_k)\bar f(x_i, x_k') + |S|\rho_n^2\E_U\! f(x_i,x_k;U) f(x_i,x_{k'};U)}{|S|^2} - {\rho_n^2\bar f(x_i, x_k)\bar f(x_i, x_k')}\\
  & = \Theta\left(\frac{\rho_n^2}{m}\right)  
\end{align*}
Then, as $\hat r_{ij} \sim (\sum_{k\in[n]\setminus\{i,j\}} \hat s_{ik}\hat s_{jk}')/(n-2)$, with $s_{jk}'$ an independent copy of $s_{jk}$, for any $k\in[n]\setminus\{i, j\}$ and conditionally on $x_i, x_j$, using the law of total variance:
\begin{align*}
\E\hat r_{ij}
  & = \rho_n^2\int_{[0,1]}\bar f(x_i, t)\bar f(x_j, t)dt,\\
\var\hat r_{ij}
  & = \big(\var(\hat s_{ik}\hat s_{jk}') + (n-2)\cov(\hat s_{ik}\hat s_{jk}', \hat s_{ik'}\hat s_{jk'}')\big)/(n-2)\\
  & = \big(\E\var(\hat s_{ik}\hat s_{jk}'\,\vert\,x_k)+\var(\rho_n^2\bar f(x_i, x_k)\bar f(x_j, x_k))\big)/(n-2) + O(\rho_n^4/m^2)\\
  & = \big(\E[\var(\hat s_{ik}\,\vert\,x_k)\var(\hat s_{jk}'\,\vert\,x_k)]+\rho_n^4\var(\bar f(x_i, x_k)\bar f(x_j, x_k))\big)/(n-2)+ O(\rho_n^4/m^2)\\
  & = \big(O(\rho_n^2/m^2)+O(\rho_n^4)\big)/(n-2)+ O(\rho_n^4/m^2)
    = O\big(\rho_n^4(n^{-1}+m^{-2}+(\rho_n^2 m^2n)^{-1})\big).
\end{align*}
Similar computation lead to $\hat\rho_n/\rho_n = 1+ O_p\big((nm)^{-1/2}\big)$. Then, as all variables are positive, we may call upon Markov's inequality, to obtain,
\[
{\hat r_{ij}}/{\hat\rho_n^2} = \int_{[0,1]}\bar f(x_i, t)\bar f(x_j, t)dt + \iota_{ij}',
\]
where $\E\iota_{ij}'=0$ amd $\iota_{ij}' = O_p(\sqrt\epsilon)$. Therefore,
\begin{align*}
(\hat r_{ii}+\hat r_{jj}-\hat r_{ij}-\hat r_{ji})/\hat\rho_n^2
& =
\int_{[0,1]}\bar f(x_i, t)^2dt+\int_{[0,1]}\bar f(x_j, t)^2dt\\
&\qquad -2\int_{[0,1]}\bar f(x_i, t)\bar f(x_j, t)dt
  +(\iota_{ii}'+\iota_{jj}'-\iota_{ij}'-\iota_{ji}')\\
& = \int_{[0,1]}\big(\bar f(x_i, t)-\bar f(x_j, t)\big)^2dt
  + \iota_{ij},
\end{align*}
where $\E\iota_{ij}=0$ and $\iota_{ij} = O_p(\sqrt\epsilon)$, which is the sought after result.
\end{proof}

\begin{proof}[Proof of~\cref{consistent-latent-ordering}]
First we note that 
\[
\frac{\widehat{\mathrm{dist}_{ij}}(A)-\widehat{\mathrm{dist}_{pq}}(A)}
  {\E\left[\widehat{\mathrm{dist}_{ij}}(A)-\widehat{\mathrm{dist}_{pq}}(A)\right]}
  = 1+\frac{\iota_{ij}-\iota_{pq}}
  {\E\left[\widehat{\mathrm{dist}_{ij}}(A)-\widehat{\mathrm{dist}_{pq}}(A)\right]}.
\]
Thus, upper bounding $\P\big[-(\iota_{ij}-\iota_{pq})>\E[\widehat{\mathrm{dist}_{ij}}(A)-\widehat{\mathrm{dist}_{pq}}(A)]\big]$ will yield the result. To produce this upper bound we will use Bernstein's equality. First, recalling the notation of the proof of~\cref{prop:consistency}, we have that
\[
\iota_{ij}-\iota_{pq} = \frac{1+O(1/n)}{n-2}\sum_{k\in[n]\setminus\{i,j,p,q\}} \left(\hat s_{ik} \hat s_{jk}'-\hat s_{pk} \hat s_{qk}'-\E[\hat s_{ik} \hat s_{jk}'-\hat s_{pk} \hat s_{qk}']\right)/\hat\rho_n^2,
\]
where the $O(1/n)$ contains the terms in $\iota_{ij}$ that implicate $p$ and $q$, and conversely, the terms $\iota_{pq}$ that implicate $i$ and $j$. Then, as the $(\hat s_{ik} \hat s_{jk}'+\hat s_{pk} \hat s_{qk}')_k$ are i.i.d. and upper bounded by $2$, and since from the proof of~\cref{prop:consistency}, $(n-2)^{-1}\sum_k\var\big(\hat s_{ik} \hat s_{jk}'/\hat\rho_n^2\big)=O(\epsilon)$, we can directly call upon Bernstein's equality to obtain for any $\nu>0$
\[
\P\big[-(\iota_{ij}+\iota_{pq})>\nu\big]\leq \exp\bigg({-\frac{\left(1+O(1/n)\right)\nu^2/2}{\frac{4}{3(n-2)}\nu\left(1+O(1/n)\right)+O(\epsilon)}}\bigg)
\]
and therefore the result.
\end{proof}

\begin{proof}[Proof of~\cref{latent-est-rate}]
Call $\bar f^\ast$ the purified graphon~\cite{lovasz12} corresponding to $\bar f$, and $J$ it's support. We assume that $J\subset [0,1]$, and write $J$ as the disjoints union of singletons and intervals in the form$J = \cup_s J_s$. Then, since by~\cite[Theorem 13.27]{lovasz12}, $\mathrm{dist}\big((\bar f^\ast,\psi(U)); \cdot, \cdot\big)$ is the metric induced by $\bar f^\ast$, embedding through $\widehat{\mathrm{dist}}(A)$ will lead to an embedding in $J$.

Conditionally on all the estimated distances being properly ordered, which will happen eventually in $n$ by~\cref{consistent-latent-ordering}, we observe that $\mathrm{dist}\big((\bar f^\ast,\psi(U)); \cdot, \cdot\big)$ first separates vertices in the $J_s$. Call $\mathcal{S}$ the set of vertices selected to be in $J_s$. Then, still within each $J_s$, 
we obtain consistency up to similarity transform and error of order $\eta_n :=\sup_{y\in J_s}\inf_{i\in\mathcal{S}}|y-\phi(x_i)|$ using~\cite[Theorem 3]{Arias2017}. Then, from~\cite{Pyke65}, and since the $x_i$'s are i.i.d. over a bounded set, we have that $\eta_n = O_p(\log n / n)$, which yields the result.
\end{proof}

\begin{proof}[Proof of~\cref{main-clt}]
First, we observe that from~\cref{latent-est-rate} and continuity of $h$, that
\begin{align*}
h(\hat x_i, \hat x_j,&\check z_l, A_{ ijl})\\
&= h(x_i, x_j, \check z_l, A_{ ijl}) + \big(h(\hat x_i, \hat x_j, \check z_l, A_{ ijl})-h( x_i, x
_j, \check z_l, A_{ ijl})\big)\\
&= h( x_i, x_j, \check z_l, A_{ ijl}) + O(\log n/n).
\end{align*}
Since $m=o\big((n/\log n)^2\big)$, we may replace the $\hat x$ by $x$ in~\eqref{main-clt-eq} and ignore the resulting error.

Then, we proceed as as~\cite[Theorem 1]{BickelLevina2012}, and observe that
\begin{align*}
h(x_i, x_j, \check z_l, A_{ ijl})
&= h(x_i, x_j, \check z_l, f(x_i, x_j,z_l))\\
&\quad + \big(h(x_i, x_j, \check z_l, A_{ ijl}) - h(x_i, x_j, \check z_l, f(x_i, x_j, \check z_l)).
\end{align*}
There, the second term is mean $0$, because $h$ is linear in its fourth argument, and by the law of total variance (conditioning by $x$ and $\check z$) and the standard CLT, is of variance $O(1/nm)$, we may ignore the error it induces in~\eqref{main-clt-eq}.

Finally, remains to prove that
\[
\frac{1}{n^2m}\sum_{i,j,l}h(x_i, x_j, \check z_l, f(x_i, x_j, \check z_l))
\]
is asymptotically normal, which is directly obtained via the CLT from two-sample U-statistics~\cite{Grams73}
which applies under our assumptions on $h$.
\end{proof}
We here present and prove a slight modification fo the above to account for replicated networks.

\begin{thm}\label{main-clt-rep}
Fix a smooth graphon function $f$. Set $h:[0,1]^3\to\mathbb{R}$ such that $h$ is symmetric in its first two arguments, linear in the third, and that $h$ and its first derivates are finite almost everywhere. Then, if $\epsilon := n^{-1}+(\rho_n^2 mn)^{-1}=o(1/\log n)$ and $m=o\big((n/\log n)^2\big)$,
\begin{equation}\label{main-clt-rep-eq}
\sqrt{n+m}\Bigg(
  \frac{1}{n^2m}\sum_{i, j}h\big(\hat x_i, \hat x_j, \bar{A}_{ij}\big)-
  \E_{u, v \sim \mathrm{U}[0,1]} h\big(u, v, \bar{f}(u, v)\big)
\Bigg)\to\mathrm{Normal}\big(0,\bar{\Sigma}\big).
\end{equation}
\end{thm}

\begin{proof}[Proof of~\cref{main-clt-rep}]
Proceeding exactly as in the proof of~\cref{main-clt}, from~\cref{latent-est-rate} and continuity of $h$, it follows that
\begin{align*}
h(\hat x_i, \hat x_j, \bar{A}_{ij})
&= h(x_i, x_j, \bar{A}_{ij}) + \big(h(\hat x_i, \hat x_j, \bar{A}_{ij})-h( x_i, x
_j, \bar{A}_{ij})\big)\\
&= h( x_i, x_j, \bar{A}_{ij}) + O(\log n/n).
\end{align*}
Since $m=o\big((n/\log n)^2\big)$, we may replace the $\hat x$ by $x$ in~\eqref{main-clt-eq} and ignore the resulting error.

Then, we proceed as as~\cite[Theorem 1]{BickelLevina2012}, and observe that 
\begin{equation*}
h(x_i, x_j, \bar{A}_{ij})
= h(x_i, x_j, \bar{f}(u,v))
 + \big(h(x_i, x_j, \bar{A}_{ij}) - h(x_i, x_j, \bar{f}(u,v)) \big).
\end{equation*}
There, the second term is mean $0$, as $h$ is linear in its third argument, and by the law of total variance (conditioning on $x$) and the standard CLT, is of variance $O(1/n)$, we may ignore the error it induces in~\eqref{main-clt-eq}.

Finally, remains to prove that: 
\[
\frac{1}{n^2m}\sum_{i,j}h(x_i, x_j, \bar{f}(u,v))
\]
is asymptotically normal, which is directly obtained via the CLT from two-sample U-statistics~\cite{Grams73}
which applies under our assumptions on $h$.
\end{proof}
\end{appendix}

\bibliographystyle{apalike}
\bibliography{NetworkSample_graphon_merged}
\end{document}